\documentclass[11pt]{article}

\usepackage[charter]{mathdesign}
\usepackage{eucal}
\usepackage{microtype}
\usepackage{authblk}

\usepackage[margin=1in,letterpaper]{geometry}

\usepackage{amsmath}
\usepackage[amsmath,thmmarks]{ntheorem}
\theoremseparator{.}
\newtheorem{thm}{Theorem}
\newtheorem{lem}[thm]{Lemma}
\newtheorem{obs}[thm]{Observation}
\newtheorem{cor}[thm]{Corollary}
\newtheorem{prop}[thm]{Proposition}
\newtheorem{rrule}[thm]{Reduction Rule}
\theoremstyle{nonumberplain}
\theoremsymbol{\quad\ensuremath{\Box}}
\theoremheaderfont{\itshape}
\theorembodyfont{\normalfont}
\newtheorem{proof}{Proof}

\usepackage{enumitem}

\usepackage{cite}
\usepackage{cleveref}
\crefname{obs}{Observation}{Observations}
\crefname{lem}{Lemma}{Lemmas}
\crefname{thm}{Theorem}{Theorems}
\crefname{cor}{Corollary}{Corollaries}
\crefname{prop}{Proposition}{Propositions}
\crefname{figure}{Figure}{Figures}
\crefname{section}{Section}{Sections}
\crefname{equation}{}{}
\crefname{rrule}{Reduction Rule}{Reduction Rules}

\usepackage{subcaption}
\usepackage{tikz}
\usetikzlibrary{positioning,backgrounds,arrows.meta,calc,intersections}
\tikzset{
  every label/.style={inner sep=0pt},
  label distance=2pt,
  vertex/.style={circle,draw=white,ultra thick,fill=black!40!white,inner sep=0pt,minimum size=6pt},
  leaf/.style={vertex,fill=red!80!black!60!white},
  subtree/.style={thick,draw=black!40!white,thin,fill=black!20!white},
  edge/.style={very thick,draw=black},
  bold edge/.style={ultra thick,draw=black},
  thin edge/.style={thick,densely dashed,draw=black},
  broken edge/.style={edge,densely dashed},
  red node/.style={fill=red!80!black!60!white},
  blue node/.style={fill=blue!80!black!60!white},
  red edge/.style={draw=red!80!black!60!white},
  blue edge/.style={draw=blue!80!black!60!white},
  in subtree/.append style={draw=black!20!white}
}

\newcommand{\dmp}{d_{\textrm{MP}}}
\newcommand{\dtbr}{d_{\textrm{TBR}}}
\newcommand{\PS}[2]{l_{#1}(#2)}
\newcommand{\ES}[2]{\Delta_{#1}(#2)}
\newcommand{\FS}[2]{\PS{#1}{#2}}
\newcommand{\NN}{\mathbb{N}}

\renewcommand{\L}{\mathcal{L}}

\newcommand{\dprime}{'\!'}
\newcommand{\trange}{\NN_{\ge 2}^\infty}
\newcommand{\poly}{\mathrm{poly}}

\newif\ifnotes
\notesfalse
\ifnotes
  \usepackage{todonotes}
  \newcommand{\nznotes}[2][]{\todo[backgroundcolor=yellow,linecolor=black,bordercolor=black,caption={},#1]{\raggedright NZ: #2}}
  \newcommand{\mjnotes}[2][]{\todo[backgroundcolor=cyan,linecolor=black,bordercolor=black,caption={},#1]{\raggedright MJ: #2}}
  \newcommand{\lvinotes}[2][]{\todo[backgroundcolor=green,linecolor=black,bordercolor=black,caption={},#1]{\raggedright LvI: #2}}
  \AddToHook{shipout/background}{\put(\paperwidth-3in,-\paperheight){\begin{tikzpicture}
    \path [fill=black!10] (0,0) rectangle (3in,\paperheight);
  \end{tikzpicture}}}
  \newcommand{\leo}[1]{\textcolor{green}{#1}}
    \newcommand{\markj}[1]{\textcolor{blue}{#1}}
    \newcommand{\nz}[1]{\textcolor{magenta}{#1}}
    
\else
  \newcommand{\nznotes}[2][]{}
  \newcommand{\mjnotes}[2][]{}
  \newcommand{\lvinotes}[2][]{}
  \newcommand{\leo}[1]{#1}
\newcommand{\markj}[1]{#1}
\newcommand{\nz}[1]{#1}

\fi

\begin{document}


\title{A Near-Linear Kernel for Two-Parsimony Distance}
\author[1]{Elise Deen}
\author[1]{Leo van Iersel\footnote{Research of Leo van Iersel and Mark Jones was partially funded by Netherlands Organization for Scientific Research (NWO) grant OCENW.KLEIN.125.}}
\author[1]{Remie Janssen}
\author[1]{Mark Jones\footnote{Corresponding author, email: m.e.l.jones@tudelft.nl}}
\author[1]{Yuki Murakami}
\author[2]{Norbert Zeh}
\affil[1]{Delft Institute of Applied Mathematics, Delft University of Technology, The Netherlands}
\affil[2]{Faculty of Computer Science, Dalhousie University, Halifax, Canada}
\maketitle

\begin{abstract}\noindent
  \markj{The maximum parsimony distance $\dmp(T_1,T_2)$ and the \leo{bounded-state} maximum parsimony distance $\dmp^t(T_1,T_2)$ measure the difference between two phylogenetic trees $T_1,T_2$ in terms of the maximum difference between their parsimony scores for any character (with $t$ a \leo{bound} on the number of states in the character, in the case of $\dmp^t(T_1,T_2)$). While computing $\dmp(T_1, T_2)$ was previously shown to be fixed-parameter tractable with a linear kernel, no such result was known for $\dmp^t(T_1,T_2)$. In this paper,}
  we prove that computing 
  $\dmp^t(T_1, T_2)$ 
  is fixed-parameter tractable \leo{for all~$t$}.
  Specifically,
  we prove that this problem has a kernel of size $O(k \lg k)$, where $k =
  \dmp^t(T_1, T_2)$. As the primary analysis tool, we introduce the concept of
  leg-disjoint incompatible quartets, which may be of independent interest.
\end{abstract}

\section{Introduction}

Parsimony \cite{fitchDefiningCourseEvolution1971} is a popular tool in
bioinformatics used to measure how closely a phylogenetic tree $T$ matches some
data associated with its leaves (e.g., DNA sequences of the taxa represented by
the leaves). Abstractly, given a labelling $f : \L(T) \rightarrow S$, called a
\emph{character}, where $\L(T)$ is the set of leaves of $T$, and $S$ is a set of
labels or \emph{states}, the goal is to extend this labelling to the internal
vertices of the tree so that the number of edges whose endpoints have different
labels is minimized.  Intuitively, these edges, called \emph{mutation edges},
reflect the number of mutation events necessary to explain the observed data
under the assumption that the tree reflects the evolution of the taxa
represented by the leaves 
\markj{(with internal vertices representing speciation events).}

Recently, Fischer and Kelk \cite{fischerMaximumParsimonyDistance2016} introduced
the maximum parsimony distance $\dmp$ as a new measure of (dis)similarity of two
phylogenetic trees $T_1$ and $T_2$ with the same leaf set $X = \L(T_1) =
\L(T_2)$. This distance is defined as the maximum difference between the parsimony scores of the two trees, where the maximum is taken over all possible
characters. Fischer and Kelk also introduced the bounded-state variant
$\dmp^t$, where the maximum is taken over all possible characters with at most
$t$ states. They proved that the problems of computing $\dmp$ and $\dmp^t$ for
$t \geq 2$ are both NP-hard \cite{fischerMaximumParsimonyDistance2016}, and that
this holds even when the trees are binary \cite{kelkComplexityComputingMP2017}. 

\leo{Appealing properties of these similarity measures include that they are related to the popular optimization criterion maximum parsimony as well as to rearrangement operations as subtree prune and regraft (SPR) and tree bisection and reconnection (TBR)~\cite{fischerMaximumParsimonyDistance2016,bruen2008parsimony}. In addition, from a computational perspective it is useful that lower bounds can easily be computed by considering a particular character. This contrasts the situation for SPR and TBR distance where upper bounds can be found by providing a sequence of SPR/TBR moves turning~$T_1$ into~$T_2$.}

For $\dmp$, which does not impose a bound on the number of states used by the
optimal character, some algorithmic results are known. Kelk and Stamoulis
\cite{kelkNoteConvexCharacters2017} gave a single-exponential algorithm for
calculating $\dmp(T_1,T_2)$ with running time $O(\phi^n \cdot \poly(n))$, where
$n$ denotes the number of taxa in $T_1$ and $T_2$, and $\phi \approx 1.618$ is
the golden ratio. Kelk et al.\ \cite{kelkReductionRulesMaximum2016} showed that
the well-known cherry and chain reduction rules (with minimum chain length 4)
are safe for $\dmp$. These rules were previously used (with chain length 3) to
give a linear kernel for  the tree bisection and reconnection (TBR) distance
$\dtbr(T_1, T_2)$ \cite{allenSubtreeTransferOperations2001}.  As observed by
Kelk et al.\ \cite{kelkReductionRulesMaximum2016}, the results of
\cite{kelkNoteConvexCharacters2017,kelkReductionRulesMaximum2016,allenSubtreeTransferOperations2001}
together imply that $\dmp$ is fixed-parameter tractable (FPT) with respect to
$\dtbr(T_1, T_2)$.  Finally, Jones, Kelk, and Stougie
\cite{jonesMaximumParsimonyDistance2021} proved that the kernel produced by the
reduction rules of Kelk et al.\ has a size that is linear also in $\dmp(T_1,
T_2)$.  The results of
\cite{kelkReductionRulesMaximum2016,jonesMaximumParsimonyDistance2021}
imply in particular that $\dtbr(T_1, T_2)$ and $\dmp(T_1, T_2)$ differ by at
most a constant factor, for any two trees $T_1$ and $T_2$ over the same set of
taxa~$X$ \markj{\cite[Theorem 5]{jonesMaximumParsimonyDistance2021}.}\lvinotes{What does ``constant'' mean here? Do there exist constants~$c,d\in\mathbb{R}$ such that for all pairs of trees~$T_1,T_2$ we have $c\cdot\dmp(T_1, T_2) \leq \dtbr(T_1, T_2) \leq d\cdot\dmp(T_1, T_2)$? MJ: Yes!}

\leo{In this paper we focus on computing $\dmp^t$, the variant of $\dmp$ where the number of states is bounded. This is arguably the most biologically relevant version since biological data usually has a bounded number of states (eg. $4$ for DNA). However, to the best of our knowledge, for $\dmp^t$ no results beyond NP-hardness were known prior to this paper.}
\markj{Our main result is}
that the maximum $t$-state parsimony distance
$\dmp^t(T_1, T_2)$ between two trees $T_1$ and $T_2$ has a near-linear kernel.
Specifically, we show that there exists a polynomial-time algorithm reducing a
pair of trees $(T_1, T_2)$ to a pair $(T_1', T_2')$ such that $\dmp^t(T_1',
T_2') = \dmp^t(T_1, T_2)$, and $T_1'$ and $T_2'$ have $O(k \lg k)$
leaves,\footnote{Throughout this paper, we use the definition that $\lg x =
\max(1,\log_2 x)$.} where $k = \dmp^t(T_1, T_2)$.  We prove this in two steps:

First, we prove that the reduction rules used by Kelk et
al.~\cite{kelkReductionRulesMaximum2016} are safe also for $\dmp^t$.  This
implies that $\dmp^t$ has a kernel for which the number of leaves $|X|$ is
linear in $\dtbr(T_1, T_2)$.  Proving this follows the same ideas used by Kelk
el al.\ but requires some care to make the arguments work for as few as two
states; the arguments used by Kelk et al.\ relied on constructing a character with
a potentially large number of states.  Moreover, our proof also implies that the
reduction rules used by Kelk et al.\ are safe for $\dmp$, but it is
significantly shorter than the original proof by Kelk et al.

Second, we prove that $\dtbr(T_1, T_2) \in O(\dmp^t(T_1, T_2) \cdot \lg |X|)$.
These two results imply that the kernel has size $O(k \lg k)$, where $k =
\dmp^t(T_1, T_2)$.  The constants in our construction are fairly large, and we
do believe that they can be improved.  This does, however, require additional
insights into how to prove that a large kernel implies that the parsimony
distance between the two trees is high.

Our proof that $\dtbr(T_1, T_2) \in \markj{O(\dmp^t(T_1, T_2) \cdot \lg |X|)}$ \mjnotes{changed from $O(k \lg k)$}
is noteworthy for two reasons.
First, while the previous result on kernelization for $\dmp$ implies that
\markj{$\dtbr(T_1, T_2) \in O(\dmp(T_1, T_2))$,}
it establishes this relationship
indirectly, via the linear kernel.  In contrast, our proof starts with an
agreement forest (AF), which provides an upper bound on the TBR distance
\cite{allenSubtreeTransferOperations2001}, and then uses this AF to construct a
large set of incompatible quartets that lead to a high parsimony distance.

Second, the proof that the kernel produced by the reduction rules by Kelk et
al.\ has size linear in $\dmp(T_1, T_2)$
\cite{jonesMaximumParsimonyDistance2021}
\markj{relied on finding pairwise \emph{disjoint} conflicting quartets between the two trees. Each such quartet contributes 1 to $\dmp(T_1,
T_2)$, so to show that $\dmp(T_1,T_2)\geq k'$ it is enough to find $k'$ pairwise disjoint conflicting quartets.}
Our key insight is that it suffices to construct a set of
incompatible quartets that satisfy a much weaker disjointness condition in one
of the two trees and can interact arbitrarily in the other tree.  Such a set of
quartets does not give a parsimony distance that is at least the number of
quartets, but the parsimony distance is \markj{still} linear in the number of quartets.  We
present a primal-dual algorithm based on an ILP formulation of the maximum
agreement forest problem \cite{werschReflectionsKernelizingComputing2020} that
finds such a set $Q$ of incompatible quartets and an AF of size $O(|Q| \cdot
\lg|X|)$.  This establishes the key claim that $\dtbr(T_1, T_2) \in
O(\dmp^t(T_1, T_2) \cdot \lg |X|)$ mentioned earlier.

The remainder of this paper is organized as follows.  \Cref{sec:preliminaries}
introduces the necessary terminology and notation, and discusses previous
results we will build upon.  \Cref{sec:reduction-rules} provides our proof that
both cherry and chain reduction are safe for $\dmp^t$.  \Cref{sec:kernel-size}
proves our bound of the size of the kernel as a function of $\dmp^t$.
\Cref{sec:conclusions} offers conclusions and a discussion of future work.

\section{Preliminaries}

\label{sec:preliminaries}

\subsection{Definitions}

\paragraph{Phylogenetic trees, induced subtrees,  restrictions, \markj{pendant subtrees,} and parents.}



Throughout this paper, a \emph{tree on $X$} is an unrooted tree with leaf set
$X$ and whose internal vertices have degree at most~$3$. When $X$ is clear from
context, we refer to a tree on $X$ simply as a \emph{tree}.  A \emph{phylogenetic
tree} on $X$ is a tree on $X$ with no vertices of degree~$2$.

Given a tree $T$ on $X$ and a subset $Y \subseteq X$, the \emph{subtree of $T$
induced by $Y$}, $T(Y)$, is the smallest subtree of $T$ that contains all leaves
in~$Y$.  The \emph{restriction of $T$ to $Y$}, $T|_Y$, is obtained from $T(Y)$
by suppressing all degree-$2$ vertices in $T(Y)$.  To \emph{suppress} a
degree-$2$ vertex $v$ with neighbours $u$ and $w$ in a tree $T$ is to remove $v$
and its incident edges from $T$ and add the edge $(u,w)$ to $T$.  The inverse
operation is to \emph{subdivide} an edge $(u,w)$ in $T$ by deleting the edge
$(u,w)$ from $T$ and adding a new vertex $v$ along with two edges $(u,v)$ and
$(v,w)$ to $T$.  \markj{Given a subset $Y\subseteq X$, a  subtree $T'$ of $T$ is
a \emph{pendant subtree of $T(Y)$ in $T$} if $T'$ and $T(Y)$ are
vertex-disjoint, there exists an edge $(u,v)$ in $T$ with $u$ a vertex of $T(Y)$
and $v$ a vertex of $T'$, and no other edge has exactly one vertex in $T'$.
Though $T$ and $T'$ are unrooted, we call $v$ the \emph{root} of the pendant
subtree $T'$.}

Every leaf $v$ of a tree $T$ has a unique neighbour, which we call the
\emph{parent} of $v$ even though $T$ is unrooted.

\paragraph{Cherries and quartets.}

A \emph{cherry} of a tree $T$ on $X$  is a pair of leaves $(a,b)$ of $T$ with the
same parent.

A \emph{quartet} of a tree $T$ on $X$ is a subset $\{a,b,c,d\} \subseteq X$ of
size~$4$.  If the path from $a$ to $b$ in $T$ is disjoint from the path from $c$
to $d$ in~$T$, then the restriction $T|_{\{a,b,c,d\}}$ of $T$ to $\{a,b,c,d\}$
has the two cherries $(a,b)$ and $(c,d)$.  We write $T|_{\{a,b,c,d\}} = ab|cd$
in this case.

A quartet $q$ is \emph{compatible} with a pair of trees $(T_1, T_2)$ on $X$ if
$T_1|_q = T_2|_q$.  Otherwise, $q$ is \emph{incompatible} with $(T_1, T_2)$.

\paragraph{Tree bisection and reconnect distance and agreement forests.}

A {tree bisection and reconnect} (TBR) operation
\cite{allenSubtreeTransferOperations2001} on a phylogenetic tree $T$ deletes an
arbitrary edge $(u,v)$ from $T$, thereby splitting $T$ into two subtrees $T_u$
and $T_v$ that contain $u$ and $v$, respectively.  It then subdivides some edge
in $T_u$ and some edge in~$T_v$, thereby creating two new vertices $u' \in T_u$
and $v' \in T_v$, and reconnects $T_u$ and $T_v$ by adding the edge $(u',v')$.
Finally, it suppresses $u$ and $v$ (which have degree $2$ after deleting the
edge $(u,v)$).  If $u$ is a leaf, then there is no edge to subdivide in~$T_u$.
In this case, we set $u' = u$ and do not suppress $u$ after adding the edge
$(u',v')$.  The case when $v$ is a leaf is handled similarly.  The \emph{TBR
distance} $\dtbr(T_1,T_2)$ between two phylogenetic trees $T_1$ and $T_2$ is the
minimum number of TBR operations necessary to transform $T_1$ into~$T_2$.

An \emph{agreement forest} (AF) of two trees $T_1$ and $T_2$ on $X$ is a partition
$F = \{X_1, \ldots, X_k\}$ of $X$ such that
\begin{itemize}[noitemsep]
  \item $T_1|_{X_i} = T_2|_{X_i}$, for all $1 \le i \le k$,
  \item $T_1(X_i)$ and $T_1(X_j)$ are edge-disjoint for all $1 \le i < j \le k$, and
  \item $T_2(X_i)$ and $T_2(X_j)$ are edge-disjoint for all $1 \le i < j \le k$.
\end{itemize}

A \emph{maximum agreement forest} (MAF) of $T_1$ and $T_2$ is an agreement
forest with the minimum number of components $X_1, \ldots, X_k$.  We refer to
this number of components as the \emph{size} $|F|$ of the forest.  It was shown
by Allen and Steel \cite{allenSubtreeTransferOperations2001} that
$\dtbr(T_1,T_2) = |F| - 1$, for any MAF $F$ of $T_1$ and $T_2$.

\paragraph{Characters, extensions, states, parsimony, and maximum parsimony distance.}

Given a tree $T$ on~$X$, a \emph{character on $X$} is a mapping $f : X \rightarrow
S$, for some non-empty set~$S$.  We call the elements of $S$ \emph{states}.  We
say that $f$ is a \emph{$t$-state character} if $|S| = t$.  Note that there is
no requirement that $f(X) = S$.

An \emph{extension} of a character $f$ on $X$ to $T$ is a labelling $\bar f:
V(T) \rightarrow S$, where $V(T)$ denotes the set of vertices of $T$, such that
$f(v) = \bar f(v)$ for every leaf $v \in X$.

A \emph{mutation edge} of $T$ with respect to some extension $\bar f$ is an
edge $(u,v)$ such that $\bar f(u) \ne \bar f(v)$.  We use $\Delta_{\bar f}(T)$
to denote the number of mutation edges of $T$ with respect to~$\bar f$.

The \emph{parsimony score} $l_f(T)$ of $T$ with respect to some character $f$ is
defined as $l_f(T) = \min_{\bar f} \Delta_{\bar f}(T)$, where the minimum is
taken over all extensions $\bar f$ of~$f$.   We call an extension $\bar f$ of
$f$ to $T$ \emph{optimal} if $\Delta_{\bar f}(T) = l_f(T)$.

The \emph{(unbounded-state) maximum parsimony distance} $\dmp(T_1,T_2)$ between
two trees on $X$ is defined as $\dmp(T_1,T_2) = \max_f |l_f(T_1) - l_f(T_2)|$,
where the maximum is taken over all characters on~$X$.  The \emph{$t$-state
maximum parsimony distance} $\dmp^t(T_1,T_2)$ between $T_1$ and $T_2$ is defined
analogously, but the maximum is taken only over all $t$-state characters on~$X$.
Throughout this paper, we refer to the unbounded-state maximum parsimony
distance $\dmp(T_1, T_2)$ as $\dmp^\infty(T_1, T_2)$ to make it explicit that it
imposes no upper bound on the number of states used by the optimal character.

All results in this paper apply to $\dmp^t$ for any $t \in \NN \cup \{\infty\}$
that satisfies $t \ge 2$.  We refer to this set of valid values of $t$ as
$\trange$.

\paragraph{Parameterized problems, kernelization, and reduction rules.}

A \emph{parameterized problem} is a language $\L \subseteq \Sigma^* \times \NN$,
where $\Sigma$ is a fixed, finite alphabet.  For an instance $(\sigma, k) \in
\Sigma^* \times \NN$, $k$ is called the \emph{parameter} of $(\sigma, k)$.  We
call $(\sigma, k)$ a \emph{yes-instance} if $(\sigma, k) \in \L$.  Otherwise,
$(\sigma, k)$ is a \emph{no-instance}.  In the case of parsimony distance, the
string $\sigma$ encodes the pair of trees $(T_1, T_2)$ and the bound $t$ on the
number of states, so we refer to an instance $(\sigma, k)$ as the instance
$(T_1, T_2, t, k)$.  This instance is a yes-instance if $\dmp^t(T_1, T_2) \le
k$.

A \emph{kernelization algorithm}, or simply \emph{kernel}, for some
parameterized problem $\L$ is a polynomial-time algorithm which given an instance
$(\sigma, k) \in \Sigma^* \times \NN$, computes another instance $(\sigma', k')
\in \Sigma^* \times \NN$ such that
\begin{itemize}[noitemsep]
  \item $(\sigma, k) \in \L$ if and only if $(\sigma', k') \in \L$,
  \item $k' \le k$, and
  \item The size $|\sigma'| + k'$ of $(\sigma', k')$ is bounded by $f(k)$, where
  $f$ is some computable function $f : \NN \rightarrow \NN$.
\end{itemize}
The function $f$ is called the \emph{size} of the kernel.

Kernels are often obtained using repeated application of reduction rules.  A
\emph{safe reduction rule} is a polynomial-time algorithm which given an
instance $(\sigma, k) \in \Sigma^* \times \NN$, computes a strictly smaller
instance $(\sigma', k') \in \Sigma^* \times \NN$ such that $(\sigma, k) \in \L$
if and only if $(\sigma', k') \in \L$.  A reduction rule comes with a condition
or conditions that need to be satisfied for this rule to be applicable.  An
instance $(\sigma, k)$ is \emph{fully reduced} with respect to a set of
reduction rules if none of these rules is applicable to $(\sigma, k)$, that is,
if $(\sigma, k)$ does not satisfy the conditions associated with any of the
reduction rules.  A kernelization algorithm based on a set of reduction rules
repeatedly applies these rules until it obtains a fully reduced instance with
respect to these rules.  This is the kernel the algorithm returns.

\subsection{Fitch's Algorithm}

An optimal extension of a character on a tree $T$ can be computed in polynomial
time using the Fitch-Hartigan
algorithm~\cite{fitchDefiningCourseEvolution1971,hartiganMinimumMutationFits1973}.
The algorithm subdivides an arbitrary edge of $T$ and uses the vertex this
introduces as the root of the tree, thereby defining a parent-child relationship
on the vertices of the tree.  The algorithm now proceeds in two phases:

The \emph{bottom-up phase} assigns a candidate set of states to every vertex of~$T$.
For a leaf $v$ with state $f(v)$, its candidate set of states is $F(v) =
\{f(v)\}$.  For an internal vertex $u$ with children $v$ and $w$, its candidate
set $F(u)$ is defined as
\begin{equation*}
  F(u) = \begin{cases}
    F(v) \cup F(w) & \text{if } F(v) \cap F(w) = \emptyset\\
    F(v) \cap F(w) & \text{if } F(v) \cap F(w) \ne \emptyset
  \end{cases}.
\end{equation*}
In the first case, we call $u$ a \emph{union vertex}.
In the second case, we call it an \emph{intersection vertex}.

The function $F : V(T) \rightarrow 2^S$ is called the \emph{Fitch map} of~$f$,
and we will refer to the set $F(v)$ associated with a vertex $v$ as $v$'s
\emph{Fitch set}.

The second, \emph{top-down phase} uses the Fitch map to compute an optimal
extension $\bar f$ of $f$ to~$T$: For the root $r$ of~$T$, we choose an
arbitrary state $\bar f(r) \in F(r)$.  For any other vertex $v$ with parent~$u$,
we choose $\bar f(v) = \bar f(u)$ if $\bar f(u) \in F(v)$.  Otherwise, we choose
$\bar f(v)$ to be an arbitrary state in $F(v)$.  \markj{Finally, we suppress the
root of $T$ that was introduced at the start of the algorithm, keeping the same
assignment of states to all other vertices. Note that this does not change the
number of mutation edges in $T$, as the root is always assigned a state that is
assigned to at least one of its children.}

We call an extension $\bar f$ computed using the Fitch-Hartigan algorithm a
\emph{Fitch extension} of~$f$.  We will also refer to an extension $\bar f$ of
$f$ to the rooted version of $T$, before suppressing $r$, as a Fitch extension
of $f$.  The meaning will be clear from context.  Note that there are optimal
extensions of $f$ that are not Fitch extensions.

\begin{lem}[Hartigan \cite{hartiganMinimumMutationFits1973}]
  \label{lem:fitch-extension}
  A Fitch extension $\bar f$ of $f$ is an optimal extension of $f$, that is,
  $\ES{\bar f}{T} = \PS{f}{T}$.  Moreover, $\ES{\bar f}{T}$ equals the number of
  union vertices in \markj{the rooted version of}\/ $T$ with respect to $f$'s
  Fitch map~$F$.
\end{lem}

\cref{lem:fitch-extension} justifies a slight overload of notation: We
use $\FS{F}{T}$ to denote the number of union vertices of $T$ with respect
to the Fitch map~$F$.  By \cref{lem:fitch-extension}, $\FS{F}{T} = \FS{f}{T}$.

\subsection{Characters on Induced Subtrees and Restricted Subtrees}

In this section, we prove some simple results on the parsimony scores and
maximum parsimony distance of restrictions of trees on $X$ to subsets of their
leaves.  These results will be used to prove that the reduction rules in
\cref{sec:reduction-rules} are safe.

\begin{lem}
  \label{lem:Induced=RestrictedPars}
  Let $T$ be a tree on $X$, let $Y \subseteq X$, and let $f$ be a character
  on~$Y$. Then $l_f(T(Y)) = l_f(T|_Y)$.
\end{lem}

\begin{proof}
  We show first that~$l_f(T(Y)) \le l_f(T|_Y)$. Consider an optimal
  extension~$\bar f$ of~$f$ to~$T|_Y$. We define an extension~$\tilde f$ of~$f$
  to~$T(Y)$ such that $\Delta_{\tilde f}(T(Y)) = \Delta_{\bar f}(T|_Y) =
  l_f(T|_Y)$. Since $l_f(T(Y)) \le \Delta_{\tilde f}(T(Y))$, it follows
  that~$l_f(T(Y)) \le l_f(T|_Y)$.
  
  By definition, $T(Y)$ can be obtained from~$T|_Y$ by subdividing edges, that
  is, by replacing edges with paths whose internal vertices have degree~$2$.
  Every vertex $v \in T|_Y$ is also a vertex of $T(Y)$.  For any such vertex
  $v$, we define $\tilde f(v) = \bar f(v)$. For every edge~$(u,v)$ of~$T|_Y$
  that is replaced by a path~$(u, p_1, \ldots, p_k, v)$ in $T(Y)$, we
  let~$\tilde f(p_i) = \bar f(u)$ for all $1 \le i \le k$. The only mutation
  edge on the path $(u,p_1, \ldots, p_k, v)$, if there is any, is the edge
  $(p_k, v)$ because $\tilde f(u) = \tilde f(p_1) = \cdots = \tilde f(p_k)$. If
  the edge $(p_k, v)$ is a mutation edge, then $\bar f(u) = \tilde f(u) = \tilde
  f(p_k) \ne \tilde f(v) = \bar f(v)$, that is, the edge $(u,v)$ also is a
  mutation edge with respect to $\bar f$.  This shows that $\Delta_{\tilde
  f}(T(Y)) = \Delta_{\bar f}(T|_Y)$.
  
  To show that~$l_f(T|_Y) \le l_f(T(Y))$, let $\tilde f$ be an optimal
  extension of $f$ to $T(Y)$. We obtain an extension $\bar f$ of $f$ to $T|_Y$
  as the restriction of $\tilde f$ to $T|_Y$. Now consider any edge $(u,v)$ of
  $T|_Y$. If $\bar f(u) \ne \bar f(v)$, then $\tilde f(u) \ne \tilde f(v)$.
  Therefore, the path $(u, p_1, \ldots, p_k, v)$ in $T(Y)$ corresponding to $(u,
  v)$ must contain at least one mutation edge. Thus, $l_f(T|_Y) \le \Delta_{\bar
  f}(T|_Y) \le \Delta_{\tilde f}(T(Y)) = l_f(T(Y))$.
\end{proof}

The next corollary follows immediately:

\begin{cor}
  \label{cor:Induced=RestrictedDist}
  Let~$T_1$ and~$T_2$ be~trees on $X$, and let~$Y \subseteq X$.  Then
  $\dmp^t(T_1|_Y, T_2|_Y) = \dmp^t(T_1(Y), T_2(Y))$ for any~$t\in \trange$.
\end{cor}

\begin{lem}
  \label{lem:restriction}
  Let $T$ be a tree on $X$, let $Y \subseteq X$, let $f$ be a character on $X$,
  and let $f'$ be the restriction of $f$ to $Y$. Then $l_{f'}(T(Y)) \le l_f(T)$.
\end{lem}

\begin{proof}
  Let $\bar f$ be an optimal extension of $f$ to $T$. The restriction of $\bar
  f$ to $T(Y)$ is an extension $\bar f'$ of $f'$ to $T(Y)$. Every mutation edge
  in $T(Y)$ with respect to $\bar f'$ is also a mutation edge in $T$ with
  respect to $\bar f$. Thus, $l_{f'}(T(Y)) \le \Delta_{\bar f'}(T(Y)) \le
  \Delta_{\bar f}(T) = l_f(T)$.
\end{proof}

Given an induced subtree $T(Y)$ and a labelling $\bar f$ of the vertices of
$T(Y)$, we define the \emph{parsimonious extension} of $\bar f$ to $T$ as the
unique labelling $\tilde f$ of the vertices in $T$ such that $\tilde f(v) = \bar
f(v)$ for all $v \in T(Y)$ and $\tilde f(v) = \bar f(w_v)$ for all $v \notin
T(Y)$, where $w_v$ is the vertex in $T(Y)$ closest to $v$.  It follows
immediately that there are no mutation edges of $T$ with respect to $\tilde f$
that do not belong to $T(Y)$, and that an edge of $T(Y)$ is a mutation edge with
respect to $\tilde f$ if and only if it is a mutation edge with respect to
$\bar f$.  Thus, we have the following observation:

\begin{obs}
  \label{obs:parsimonious-extension}
  If $T$ is a tree on $X$, $Y \subseteq X$, $\bar f$ is a labelling of the
  vertices of~$T(Y)$, and $\tilde f$ is the parsimonious extension of $\bar f$
  to~$T$, then $\Delta_{\bar f}(T(Y)) = \Delta_{\tilde f}(T)$.
\end{obs}

\begin{lem}
  \label{lem:subset-restriction}
  Let $T_1$ and $T_2$ be trees on $X$, and let $Y \subseteq X$.
  Then $\dmp^t(T_1(Y), T_2(Y)) = \dmp^t(T_1|_Y, T_2|_Y) \leq \dmp^t(T_1, T_2)$,
  for any $t\in \trange$.
\end{lem}

\begin{proof}
  \cref{cor:Induced=RestrictedDist} states that $\dmp^t(T_1(Y), T_2(Y)) =
  \dmp^t(T_1|_Y, T_2|_Y)$. Thus, it suffices to prove that $\dmp^t(T_1(Y),
  T_2(Y)) \leq \dmp^t(T_1, T_2)$.
  
  Let $f$ be a $t$-state character on $Y$ such that $\dmp^t(T_1(Y), T_2(Y)) =
  |l_f(T_1(Y)) - l_f(T_2(Y))|$. Moreover, assume that $l_f(T_1(Y)) \le
  l_f(T_2(Y))$, so $\dmp^t(T_1(Y), T_2(Y)) = l_f(T_2(Y)) - l_f(T_1(Y))$. Let
  $\bar f$ be an optimal extension of $f$ to $T_1(Y)$, let $\tilde f$ be the
  parsimonious extension of $\bar f$ to $T_1$, and let $f'$ be the restriction
  of $\tilde f$ to the leaves of $T_1$ (i.e., to $X$).
  \markj{Thus, $f$ is the restriction of $f'$ to $Y$.} Then $l_{f'}(T_1) \le
  \Delta_{\tilde f}(T_1) = \Delta_{\bar f}(T_1(Y)) = l_f(T_1(Y))$, by
  \cref{obs:parsimonious-extension}.  By~\cref{lem:restriction}, we have
  $l_f(T_2(Y)) \le l_{f'}(T_2)$.  Thus, $\dmp^t(T_1, T_2) \ge l_{f'}(T_2) -
  l_{f'}(T_1) \ge l_f(T_2(Y)) - l_f(T_1(Y)) = \dmp^t(T_1(Y), T_2(Y))$.
\end{proof}

\section{Reduction Rules}

\label{sec:reduction-rules}

Previous kernelization results for SPR distance
\cite{bordewichComputationalComplexityRooted2005}, TBR distance
\cite{allenSubtreeTransferOperations2001}, and hybridization number
\cite{bordewichComputingHybridizationNumber2007} employ two simple reduction
rules: cherry reduction and chain reduction (see below).  It was shown that
these rules produce kernels of size linear in the SPR distance
\cite{bordewichComputationalComplexityRooted2005}, TBR distance
\cite{allenSubtreeTransferOperations2001} or hybridization number
\cite{bordewichComputingHybridizationNumber2007} of the two input
trees.\footnote{Technically, for hybridization number, the reduction rules by
Bordewich and Semple \cite{bordewichComputingHybridizationNumber2007} do not
yield a linear kernel but a linear \emph{compression}: the reduction takes an
instance of the maximum acyclic agreement forest (MAAF) problem, which is
equivalent to hybridization number, and produces an instance of linear size of a
\emph{weighted} version of the MAAF problem.}  Kelk et al.\
\cite{kelkReductionRulesMaximum2016} proved that these rules are safe also for
the unbounded-state maximum parsimony distance, as long as chain reduction is
applied only to chains of length greater than~$4$.  Jones, Kelk, and Stougie
\cite{jonesMaximumParsimonyDistance2021} proved that these reduction rules once
again produce a kernel of size linear in $\dmp^\infty(T_1, T_2)$.

In this section, we prove that cherry reduction and chain reduction are safe
also for the $t$-state parsimony distance, for any $t \ge 2$, again as long as
we apply chain reduction only to chains of length greater than~$4$.  This shows
that there exists a linear-size kernel for $\dmp^t(T_1, T_2)$ parameterized by
the TBR distance $\dtbr(T_1, T_2)$, as summarized in the following theorem:

\begin{thm}
  \label{thm:reduction-rules}
  There exists a set of \markj{safe} reduction rules for $\dmp^t$ such that the
  two trees $T_1$ and $T_2$ on $X$ in a fully reduced instance\nznotes{This said
  ``yes-instance'' before, but whether we have a yes- or no-instance is
  irrelevant for the bound on the size of the kernel.} $(T_1, T_2, t, k)$
  satisfy $|X| \le 20 \cdot \dtbr(T_1, T_2)$, \markj{for any $t\in \trange$.}
\end{thm}

Kelk and Linz proved that a fully reduced instance with respect to cherry
reduction and chain reduction applied to chains of length greater than $3$ has
size at most~$15 \cdot \dtbr(T_1, T_2) - 9$ \cite{kelkTightKernelComputing2019}.
Since we apply chain reduction only to chains of length greater than~$4$, we
obtain a kernel that is up to a factor of $\frac{4}{3}$ bigger, which gives the
bound of $\frac{4}{3} \cdot 15 \cdot \dtbr(T_1, T_2) = 20 \cdot \dtbr(T_1, T_2)$
on the size of a fully reduced instance in \cref{thm:reduction-rules}.

\subsection{Cherry Reduction}

Cherry reduction eliminates common cherries of the two input trees:

\begin{rrule}[Cherry Reduction]
  \label{rule:cherry-reduction}
  If $T_1$ and $T_2$ have a common cherry $(x,y)$, that is, if $x$ and $y$ are
  two leaves that have the same parent in both $T_1$ and $T_2$, then remove $y$
  from both $T_1$ and $T_2$ and suppress the parent of $x$ and~$y$.
\end{rrule}

Another way to state cherry reduction is that we replace $T_1$ and $T_2$ with
their restrictions to $X \setminus \{y\}$.  The following lemma shows that
applying cherry reduction to a pair of trees $(T_1, T_2)$ does not change their
maximum parsimony distance.

\begin{lem}
  \label{lem:cherry-reduction-lower-bound}
  If $T_1' = T_1|_{X \setminus \{y\}}$ and $T_2' = T_2|_{X \setminus \{y\}}$ are
  the two trees obtained from $T_1$ and $T_2$ by applying cherry reduction to a
  common cherry $(x,y)$ of $T_1$ and~$T_2$, then \markj{$\dmp^t(T_1', T_2') =
  \dmp^t(T_1, T_2)$,} for any $t \in \trange$.
\end{lem}

\begin{proof}
  By Lemma~\ref{lem:subset-restriction}, we have that $\dmp^t(T_1', T_2')
  \le \dmp^t(T_1, T_2)$. Therefore, it is sufficient to show that $\dmp^t(T_1',
  T_2') \ge \dmp^t(T_1, T_2)$.

  Since $x$ and $y$ have the same parent in both $T_1$ and $T_2$, we use $p$ to
  refer to this common parent in both $T_1$ and $T_2$, considering it the same
  vertex whether it belongs to $T_1$ or~$T_2$. Similarly, we consider the third
  neighbour of $p$ in $T_1$ and $T_2$ to be the same vertex~$q$. In other words,
  the neighbourhood of $p$ in both $T_1$ and $T_2$ is $\{x, y, q\}$. In $T_1'$
  and $T_2'$, $y$ and $p$ are removed, and $q$ becomes $x$'s parent. We argue
  about rooted versions of $T_1$, $T_2$, $T_1'$, and $T_2'$. In $T_1$ and $T_2$,
  we subdivide the edge $(p,q)$ using a new vertex $r$, and we make $r$ the root
  of $T_1$ and $T_2$. After pruning $y$ and suppressing $p$, this results in $r$
  being $x$'s parent in $T_1'$ and~$T_2'$.  See \cref{fig:cherry-reduction}.

  \begin{figure}[h]
    \centering
    \subcaptionbox{\label{fig:cherry-unreduced}}{\begin{tikzpicture}
      \node [vertex,label=below:$x$] at (0,0)  (x) {};
      \node [vertex,label=below:$y$] at (1,0)  (y) {};
      \node [vertex,label=left:$p$] at (60:1) (p) {};
      \path (p) +(0:2) node [vertex,label=right:$q$] (q) {} ++(60:1) +(0:0.5) node [vertex,label=above:$r$] (r) {};
      \begin{scope}[on background layer]
        \path [subtree] (q.center) -- +(240:2) -- +(300:2) -- cycle;
      \end{scope}
      \draw [edge] (x) -- (p) -- (y) (p) -- (r) -- (q);
    \end{tikzpicture}}%
    \hspace{1in}%
    \subcaptionbox{\label{fig:cherry-reduced}}{\begin{tikzpicture}
      \node [vertex,label=below:$x$] at (0,0)  (x) {};
      \path (x) +(0:2) node [vertex,label=right:$q$] (q) {} ++(60:1) +(0:0.5) node [vertex,label=above:$r$] (r) {};
      \begin{scope}[on background layer]
        \path [subtree] (q.center) -- +(240:2) -- +(300:2) -- cycle;
      \end{scope}
      \draw [edge] (x) -- (r) -- (q);
    \end{tikzpicture}}
    \caption{The rooted version of the tree $T_1$ or $T_2$ in the input to
    cherry reduction (\subref{fig:cherry-unreduced}) and the corresponding reduced
    tree $T_1'$ or $T_2'$ (\subref{fig:cherry-reduced}).}%
    \label{fig:cherry-reduction}
  \end{figure}
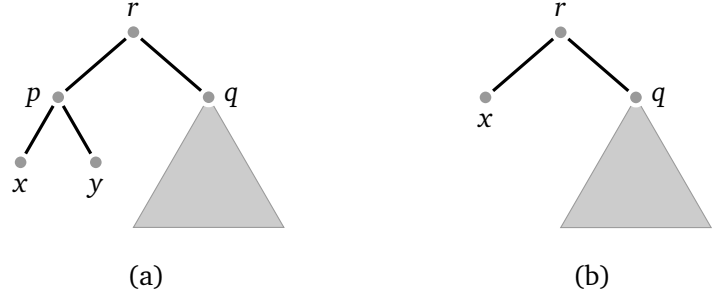

  Let $f$ be a $t$-state character on $X$ with Fitch maps $F_1$
  and $F_2$ on $T_1$ and $T_2$ such that
  \begin{equation*}
    |\FS{F_1}{T_1} - \FS{F_2}{T_2}| = \dmp^t(T_1, T_2).
  \end{equation*}
  Without loss of generality, assume that $\FS{F_1}{T_1} \le \FS{F_2}{T_2}$, so
  $\dmp^t(T_1,T_2) = \PS{F_2}{T_2} - \PS{F_1}{T_1}$.  To prove that
  $\dmp^t(T_1', T_2') \ge \dmp^t(T_1, T_2)$, we construct a $t$-state character
  $f'$ on $X' = X \setminus \{y\}$ whose Fitch maps $F_1'$ and $F_2'$ on $T_1'$
  and $T_2'$ satisfy
  \begin{equation*}
    \FS{F_2'}{T_2'} - \FS{F_1'}{T_1'} \ge \FS{F_2}{T_2} -
    \FS{F_1}{T_1}.
  \end{equation*}
  This implies that
  \begin{equation*}
    \dmp^t(T_1', T_2') \ge \FS{F_2'}{T_2'} - \FS{F_1'}{T_1'} \ge
    \FS{F_2}{T_2} - \FS{F_1}{T_1} = \dmp^t(T_1, T_2).
  \end{equation*}

  We define $f'$ by choosing $f'(z) = f(z)$ for all $z \ne x$. We choose $f'(x)$
  arbitrarily from $F_1(p) \cap F_1(q)$ if $r$ is an intersection vertex
  in~$T_1$. Otherwise, we choose $f'(x) = f(x)$.  For every vertex $z \in T_i'$
  such that $z \notin \{x, r\}$, we have $F_i'(z) = F_i(z)$ because any such
  vertex has the same set of descendant leaves in $T_i$ and $T_i'$ and any such
  descendant leaf $z'$ satisfies $f'(z') = f(z')$ (see
  \cref{fig:cherry-reduction}).  Therefore, every vertex $z \ne r$ is a union
  vertex in $T_i'$ if and only if it is a union vertex in~$T_i$.  Moreover, $p$
  is a union vertex in $T_1$ if and only if it is a union vertex in~$T_2$.
  Thus,
  \begin{equation*}
    \FS{F_2'}{T_2'} - \FS{F_1'}{T_1'}
    = \FS{F_2}{T_2} - \FS{F_1}{T_1} + (u_2' - u_2) - (u_1' - u_1),
  \end{equation*}
  where
  \begin{equation*}
    u_i = \begin{cases}
      1 & \text{if $r$ is a union vertex in $T_i$}\\
      0 & \text{otherwise}
    \end{cases}
  \end{equation*}
  and
  \begin{equation*}
    u_i' = \begin{cases}
      1 & \text{if $r$ is a union vertex in $T_i'$}\\
      0 & \text{otherwise,}
    \end{cases}
  \end{equation*}
  for $i \in \{1, 2\}$.  Therefore,
  \begin{equation*}
    \FS{F_2'}{T_2'} - \FS{F_1'}{T_1'} \ge \FS{F_2}{T_2} - \FS{F_1}{T_1}
  \end{equation*}
  if and only if
  \begin{equation*}
    u_2' - u_2 \ge u_1' - u_1.
  \end{equation*}

  To prove that this inequality holds, note first that the choice of $f'(x)$
  when $r$ is an intersection vertex in $T_1$ ensures that $r$ is an
  intersection vertex also in~$T_1'$.  Thus, $u_1' - u_1 \le 0$.

  Next observe that no matter whether $f'(x)$ is chosen from $F_1(p) \cap
  F_1(q)$ or $f'(x) = f(x)$, we have $f'(x) \in F_1(p) = F_2(p)$.  Thus, if $r$
  is a union vertex in~$T_2$, it is also a union vertex in~$T_2'$.  Therefore,
  $u_2' - u_2 \ge 0$.

  Since $u_2' - u_2 \ge 0$ and $u_1' - u_1 \le 0$, we have $u_2' - u_2 \ge u_1'
  - u_1$, as desired.
\end{proof}

\subsection{Chain Reduction}

A \emph{chain of length $k$} in a tree $T$ is an ordered sequence of leaves
$\langle x_1, \ldots, x_k \rangle$ such that $\langle p_1, \ldots, p_k \rangle$
is a path in~$T$, where $p_i$ is the parent of $x_i$ in $T$, for all $1 \le i
\le k$.  It is possible to have $p_1 = p_2$ and/or $p_{k-1} = p_k$. If this is
the case, the chain is called \emph{pendant} in~$T$. A \emph{common chain} of
$T_1$ and $T_2$ is an ordered sequence of leaves $\langle x_1, \ldots, x_k
\rangle$ that is a chain in both $T_1$ and~$T_2$.  Chain reduction ensures that
the trees in a fully reduced instance do not have long common chains:

\begin{rrule}[Chain Reduction]
  \label{rule:chain-reduction}
  If $T_1$ and $T_2$ have a common chain $\langle x_1,x_2,\ldots, x_k \rangle$
  of length $k \ge 5$, then remove the leaves $x_3, \ldots, x_{k-2}$\, from both
  $T_1$ and~$T_2$, and suppress their parents in both trees.
\end{rrule}

It was shown by Kelk et al.\ \cite{kelkReductionRulesMaximum2016} that chain
reduction preserves $\dmp^\infty(T_1, T_2)$.  The argument by Kelk et al.\ uses
what was called a \emph{less constrained roots argument} in that paper.  Here we
extend this argument to bounded-state characters to prove that chain reduction
also preserves $\dmp^t(T_1, T_2)$, for any finite~$t$.

The tree $T_1$ consists of the chain $\langle x_1, \ldots, x_k \rangle$ plus two
pendant subtrees $T_A$ and $T_B$ whose roots are adjacent to $p_1$ and $p_k$,
respectively.  See \cref{fig:chain-reduction-unrooted}.  If $\langle x_1,
\ldots, x_k \rangle$ is a pendant chain of $T_1$, then $T_A$ or $T_B$ is empty,
possibly both. Similarly, $T_2$ consists of the chain $(x_1, \ldots, x_k)$ plus
two pendant subtrees $T_C$ and $T_D$ whose roots are adjacent to $p_1$ and
$p_k$, respectively. Again, $T_C$ or $T_D$ may be empty, possibly both.  For $P
\in \{A, B, C, D\}$, let $X_P$ be the set of leaves in $T_P$, let $r_P$ be the
root of $T_P$, and let $e_P$ be the edge connecting $r_P$ to $p_1$ or $p_k$.
Obviously, $X_P = \emptyset$, and $r_P$ and $e_P$ do not exist, if $T_P$ is
empty.  Note that $X_A \cup X_B = X_C \cup X_D = X \setminus \{x_1, \ldots,
x_k\}$.

\begin{figure}[t]
  \begin{tikzpicture}
    \node [vertex,label=above:$r_A$] at (-1,0) (rA) {};
    \node [vertex,label=above:$p_1$] at (0,0) (p1) {};
    \node [vertex,label=above:$p_2$] at (0.5,0) (p2) {};
    \node [vertex,label={[xshift=-3pt]above:$p_{k-1}$}] at (2,0) (pk-1) {};
    \node [vertex,label=above:$p_k$] at (2.5,0) (pk) {};
    \node [vertex,label=above:$r_B$] at (3.5,0) (rB) {};
    \node [vertex,label=below:$x_1$] at (0,-1) (x1) {};
    \node [vertex,label=below:$x_2$] at (0.5,-1) (x2) {};
    \node [vertex,label={[xshift=-3pt]below:$x_{k-1}$}] at (2,-1) (xk-1) {};
    \node [vertex,label=below:$x_k$] at (2.5,-1) (xk) {};
    \node at (-1.9,0) {$T_A$};
    \node at (4.4,0) {$T_B$};
    \begin{scope}[on background layer]
      \path [subtree] (rA.center) -- +(210:1.5) -- +(150:1.5) -- cycle;
      \path [subtree] (rB.center) -- +(30:1.5) -- +(330:1.5) -- cycle;
    \end{scope}
    \draw [edge] (x1) -- (p1) -- (p2) -- (x2) (xk-1) -- (pk-1) -- (pk) -- (xk)
    (rA) to node [below] {$e_A$} (p1) 
    (pk) to node [below] {$e_B$} (rB);
    \draw [edge] (p2) -- (barycentric cs:p2=0.685,pk-1=0.315);
    \draw [edge] (pk-1) -- (barycentric cs:p2=0.315,pk-1=0.685);
    \draw [edge,dotted] (barycentric cs:p2=0.315,pk-1=0.685) -- (barycentric cs:p2=0.685,pk-1=0.315);
    \node [anchor=north,yshift=-5mm] at (current bounding box.south) {$T_1$};
  \end{tikzpicture}%
  \hspace{\stretch{1}}%
  \begin{tikzpicture}
    \node [vertex,label=above:$r_C$] at (-1,0) (rA) {};
    \node [vertex,label=above:$p_1$] at (0,0) (p1) {};
    \node [vertex,label=above:$p_2$] at (0.5,0) (p2) {};
    \node [vertex,label={[xshift=-3pt]above:$p_{k-1}$}] at (2,0) (pk-1) {};
    \node [vertex,label=above:$p_k$] at (2.5,0) (pk) {};
    \node [vertex,label=above:$r_D$] at (3.5,0) (rB) {};
    \node [vertex,label=below:$x_1$] at (0,-1) (x1) {};
    \node [vertex,label=below:$x_2$] at (0.5,-1) (x2) {};
    \node [vertex,label={[xshift=-3pt]below:$x_{k-1}$}] at (2,-1) (xk-1) {};
    \node [vertex,label=below:$x_k$] at (2.5,-1) (xk) {};
    \node at (-1.9,0) {$T_C$};
    \node at (4.4,0) {$T_D$};
    \begin{scope}[on background layer]
      \path [subtree] (rA.center) -- +(210:1.5) -- +(150:1.5) -- cycle;
      \path [subtree] (rB.center) -- +(30:1.5) -- +(330:1.5) -- cycle;
    \end{scope}
    \draw [edge] (x1) -- (p1) -- (p2) -- (x2) (xk-1) -- (pk-1) -- (pk) -- (xk)
    (rA) to node [below] {$e_C$} (p1) 
    (pk) to node [below] {$e_D$} (rB);
    \draw [edge] (p2) -- (barycentric cs:p2=0.685,pk-1=0.315);
    \draw [edge] (pk-1) -- (barycentric cs:p2=0.315,pk-1=0.685);
    \draw [edge,dotted] (barycentric cs:p2=0.315,pk-1=0.685) -- (barycentric cs:p2=0.685,pk-1=0.315);
    \node [anchor=north,yshift=-5mm] at (current bounding box.south) {$T_2$};
  \end{tikzpicture}
  \caption{The various definitions of vertices, edges, and subtrees in the discussion of chain reduction.}
  \label{fig:chain-reduction-unrooted}
\end{figure}
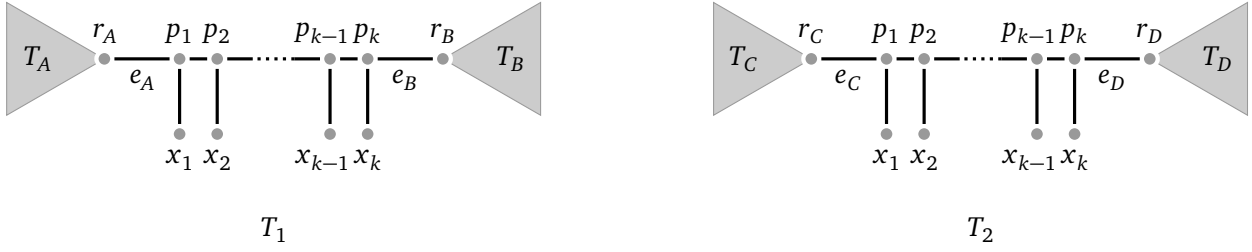
\begin{figure}[t]
  \hspace{\stretch{1}}%
  \begin{tikzpicture}
    \path       node [vertex,label=left:$r_A$]           (rA)   {}
    ++(60:1)    node [vertex,label=left:$p_1$]           (p1)   {}
    +(300:1)    node [vertex,label=below:$x_1$]          (x1)   {}
    ++(60:1)    node [vertex,label=left:$p_2$]           (p2)   {}
    +(300:1)    node [vertex,label=below:$x_2$]          (x2)   {}
    +(60:1)     node [vertex,xshift=5mm,label=above:$r$] (r)    {}
    ++(0:2)     node [vertex,label=right:$p_3$]          (p3)   {}
    +(240:1)    node [vertex,label=below:$x_3$]          (x3)   {}
    ++(300:1.5) node [vertex,label=right:$p_{k-1}$]      (pk-1) {}
    +(240:1)    node [vertex,label=below:$x_{k-1}$]      (xk-1) {}
    ++(300:1)   node [vertex,label=right:$p_k$]          (pk)   {}
    +(240:1)    node [vertex,label=below:$x_x$]          (xk)   {}
    ++(300:1)   node [vertex,label=right:$r_B$]          (rB)   {};
    \path (rA) +(270:0.9) node {$T_A$};
    \path (rB) +(270:0.9) node {$T_B$};
    \begin{scope}[on background layer]
      \path [subtree] (rA.center) -- +(240:1.5) -- +(300:1.5) -- cycle;
      \path [subtree] (rB.center) -- +(240:1.5) -- +(300:1.5) -- cycle;
    \end{scope}
    \draw [edge] (rA) to node [left,xshift=-2pt] {$e_A$} (p1) -- (x1) (p1) -- (p2) -- (x2) (p2) -- (r) -- (p3) -- (x3)
    (xk-1) -- (pk-1) -- (pk) (xk) -- (pk) to node [right,xshift=1pt] {$e_B$} (rB);
    \draw [edge] (p3) -- (barycentric cs:p3=0.685,pk-1=0.315);
    \draw [edge] (pk-1) -- (barycentric cs:p3=0.315,pk-1=0.685);
    \draw [edge,dotted] (barycentric cs:p3=0.315,pk-1=0.685) -- (barycentric cs:p3=0.685,pk-1=0.315);
    \node [anchor=north,yshift=-5mm] at (current bounding box.south) {$T_1$};
  \end{tikzpicture}%
  \hspace{\stretch{2}}%
  \begin{tikzpicture}
    \path       node [vertex,label=left:$r_C$]           (rA)   {}
    ++(60:1)    node [vertex,label=left:$p_1$]           (p1)   {}
    +(300:1)    node [vertex,label=below:$x_1$]          (x1)   {}
    ++(60:1)    node [vertex,label=left:$p_2$]           (p2)   {}
    +(300:1)    node [vertex,label=below:$x_2$]          (x2)   {}
    +(60:1)     node [vertex,xshift=5mm,label=above:$r$] (r)    {}
    ++(0:2)     node [vertex,label=right:$p_3$]          (p3)   {}
    +(240:1)    node [vertex,label=below:$x_3$]          (x3)   {}
    ++(300:1.5) node [vertex,label=right:$p_{k-1}$]      (pk-1) {}
    +(240:1)    node [vertex,label=below:$x_{k-1}$]      (xk-1) {}
    ++(300:1)   node [vertex,label=right:$p_k$]          (pk)   {}
    +(240:1)    node [vertex,label=below:$x_x$]          (xk)   {}
    ++(300:1)   node [vertex,label=right:$r_D$]          (rB)   {};
    \path (rA) +(270:0.9) node {$T_C$};
    \path (rB) +(270:0.9) node {$T_D$};
    \begin{scope}[on background layer]
      \path [subtree] (rA.center) -- +(240:1.5) -- +(300:1.5) -- cycle;
      \path [subtree] (rB.center) -- +(240:1.5) -- +(300:1.5) -- cycle;
    \end{scope}
    \draw [edge] (rA) to node [left,xshift=-2pt] {$e_C$} (p1) -- (x1) (p1) -- (p2) -- (x2) (p2) -- (r) -- (p3) -- (x3)
    (xk-1) -- (pk-1) -- (pk) (xk) -- (pk) to node [right,xshift=1pt] {$e_D$} (rB);
    \draw [edge] (p3) -- (barycentric cs:p3=0.685,pk-1=0.315);
    \draw [edge] (pk-1) -- (barycentric cs:p3=0.315,pk-1=0.685);
    \draw [edge,dotted] (barycentric cs:p3=0.315,pk-1=0.685) -- (barycentric cs:p3=0.685,pk-1=0.315);
    \node [anchor=north,yshift=-5mm] at (current bounding box.south) {$T_2$};
  \end{tikzpicture}%
  \hspace{\stretch{1}}
  \caption{The rooted versions of $T_1$ and $T_2$ considered for the discussion of chain reduction.}
  \label{fig:chain-reduction-rooted}
\end{figure}
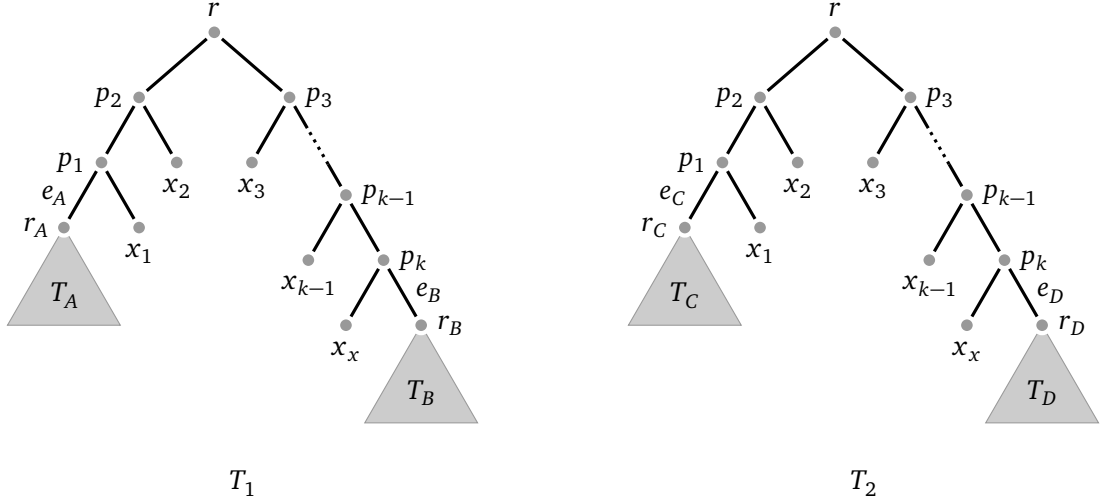

Throughout the remainder of this subsection, we consider rooted versions of
$T_1$ and $T_2$, where the root $r$ is placed on the edge $(p_2, p_3)$ in both
trees.  See \cref{fig:chain-reduction-rooted}.  This makes $r$ the common parent
of $p_2$ and $p_{k-1}$ in the two trees obtained from $T_1$ and $T_2$ by
applying chain reduction.  Now fix an optimal $t$-state character $f : X
\rightarrow S$, that is, a $t$-state character such that $\dmp^t(T_1, T_2) =
|l_f(T_1) - l_f(T_2)|$, and assume w.l.o.g.\ that $l_f(T_1) \le l_f(T_2)$, so
$\dmp^t(T_1, T_2) = l_f(T_2) - l_f(T_1)$.  For $i \in \{1, 2\}$, let $F_i$ be
the Fitch map of $T_i$ defined by~$f$, and let $u_i$ be the the number of union
vertices among $r, p_1, \ldots, p_k$ in~$T_i$.  Finally, let $f_P$ be the
restriction of $f$ to $T_P$ and let $S_P$ be the Fitch set of $r_P$ defined by
the character~$f_P$, for $P \in \{A, B, C, D\}$. If $T_P$ is empty, then let
$S_P = S$. Then
\begin{align*}
    l_f(T_1) &= l_{f_A}(T_A) + l_{f_B}(T_B) + u_1,\\
    l_f(T_2) &= l_{f_C}(T_C) + l_{f_D}(T_D) + u_2.
\end{align*}
In particular,
\begin{equation} \label{eq3.2}
   \dmp^t(T_1, T_2) = l_{f_C}(T_C) + l_{f_D}(T_D) - l_{f_A}(T_A) - l_{f_B}(T_B) + u_2 - u_1.
\end{equation}

The less constrained roots argument now bounds the difference $u_2 - u_1$
depending on whether $S_A \subseteq S_C$ and $S_B \subseteq S_D$. The name
refers to the fact that, for example, $S_A \subseteq S_C$ implies that the
choice of the state $\bar f_2(r_C)$ in a Fitch extension $\bar f_2$ of $f$ to
$T_2$ is less constrained than the choice of $\bar f_1(r_A)$ in a Fitch
extension $\bar f_1$ of $f$ to~$T_1$.

\begin{lem}
  \label{lemma:lcra}
  Let $\delta_{AC} = 0$ if $S_A \subseteq S_C$, and $\delta_{AC} = 1$ otherwise.
  Similarly, let $\delta_{BD} = 0$ if $S_B \subseteq S_D$, and $\delta_{BD} = 1$
  otherwise.  Then $u_2 - u_1 \le \delta_{AC} + \delta_{BD}$.
\end{lem}

\begin{proof}
  Suppose first that all four subtrees $T_A,T_B,T_C,T_D$ are non-empty, and
  consider a Fitch extension $\bar f_1$ of $f$ to~$T_1$. Then $\bar f_1(r_A) = a
  \in S_A$, and $\bar f_1(r_B) = b \in S_B$.  We construct an extension $\bar
  f_2$ of $f$ to $T_2$ as follows: 
  
  We start by setting $\bar f_2(v) = f(v)$ for
  every leaf $v \in X$.
  To define the labels of all internal vertices of~$T_2$, we pick states $c \in
  S_C$ and $d \in S_D$, and set $\bar f_2(r_C) = c$ and $\bar f_2(r_D) = d$. If
  $S_A \subseteq S_C$, then we choose $c = a$. If $S_B \subseteq S_D$, then we
  choose $d = b$. We label the remaining vertices in $T_C$ and $T_D$ so that the
  restriction of $\bar f_2$ to $T_C$ is a Fitch extension of~$f_C$, and the
  restriction of $\bar f_2$ to $T_D$ is a Fitch extension of~$f_D$. Finally, we
  complete $\bar f_2$ by setting $\bar f_2(p_i) = \bar f_1(p_i)$ for all $1 \le
  i \le k$.\footnote{Note that this is well defined because $T_A,T_B,T_C,T_D \ne
  \emptyset$ implies that no two leaves $x_i$ and $x_j$ have the same parent in
  either $T_1$ or~$T_2$.}

  This ensures that $T_1(\{x_1, \ldots, x_k\})$ and $T_2(\{x_1, \ldots, x_k\})$
  contain the same mutation edges with respect to $\bar f_1$ and $\bar f_2$,
  respectively.  The edge $e_C$ is a mutation edge only if $e_A$ is or $S_A
  \not\subseteq S_C$.  The edge $e_D$ is a mutation edge only if $e_B$ is or
  $S_B \not\subseteq S_D$.  Since the restrictions of $\bar f_1$ and $\bar f_2$
  to $T_A$, $T_B$, $T_C$, and $T_D$ are Fitch extensions of $f_A$, $f_B$, $f_C$,
  and~$f_D$, this shows that
  \begin{equation*}
    \Delta_{\bar f_2}(T_2) - \Delta_{\bar f_1}(T_1) \le
    l_{f_C}(T_C) + l_{f_D}(T_D) - l_{f_A}(T_A) - l_{f_B}(T_B) + \delta_{AC} + \delta_{BD}.
  \end{equation*}
  Since $\bar f_2$ is an extension of $f$ to $T_2$ and $\bar f_1$ is a Fitch
  extension of $f$ to $T_1$, we also have $l_f(T_2) \le \Delta_{\bar f_2}(T_2)$
  and $l_f(T_1) = \Delta_{\bar f_1}(T_1)$. Thus,
  \begin{equation*}
    \dmp^t(T_1, T_2) = l_f(T_2) - l_f(T_1) \le
    l_{f_C}(T_C) + l_{f_D}(T_D) - l_{f_A}(T_A) - l_{f_B}(T_B) + \delta_{AC} + \delta_{BD}.
  \end{equation*}
  Together with \cref{eq3.2}, this shows that
  \begin{equation*}
    u_2 - u_1 \le \delta_{AC} + \delta_{BD}.
  \end{equation*}

  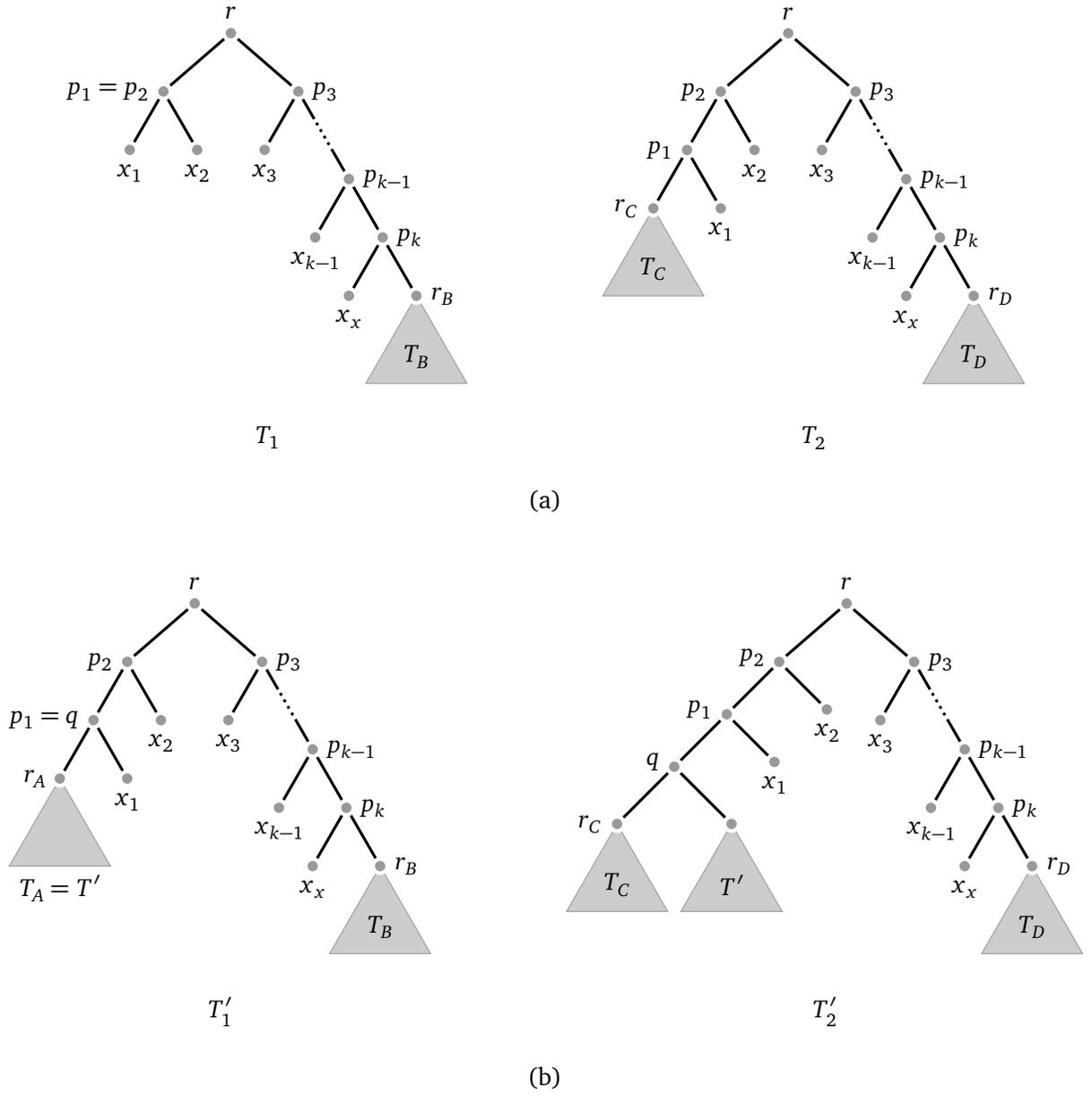
\begin{figure}[p]
    \centering
    \subcaptionbox{}{\begin{tikzpicture}
      \path       node [vertex,label=below:$x_1$]          (x1)   {}
      ++(60:1)    node [vertex,label=left:{$p_1 = p_2$}]   (p2)   {}
      +(300:1)    node [vertex,label=below:$x_2$]          (x2)   {}
      +(60:1)     node [vertex,xshift=5mm,label=above:$r$] (r)    {}
      ++(0:2)     node [vertex,label=right:$p_3$]          (p3)   {}
      +(240:1)    node [vertex,label=below:$x_3$]          (x3)   {}
      ++(300:1.5) node [vertex,label=right:$p_{k-1}$]      (pk-1) {}
      +(240:1)    node [vertex,label=below:$x_{k-1}$]      (xk-1) {}
      ++(300:1)   node [vertex,label=right:$p_k$]          (pk)   {}
      +(240:1)    node [vertex,label=below:$x_x$]          (xk)   {}
      ++(300:1)   node [vertex,label=right:$r_B$]          (rB)   {};
      \path (rB) +(270:0.9) node {$T_B$};
      \begin{scope}[on background layer]
        \path [subtree] (rB.center) -- +(240:1.5) -- +(300:1.5) -- cycle;
      \end{scope}
      \draw [edge] (x1) -- (p2) -- (x2) (p2) -- (r) -- (p3) -- (x3)
      (xk-1) -- (pk-1) -- (pk) (xk) -- (pk) -- (rB);
      \draw [edge] (p3) -- (barycentric cs:p3=0.685,pk-1=0.315);
      \draw [edge] (pk-1) -- (barycentric cs:p3=0.315,pk-1=0.685);
      \draw [edge,dotted] (barycentric cs:p3=0.315,pk-1=0.685) -- (barycentric cs:p3=0.685,pk-1=0.315);
      \node [anchor=north,yshift=-5mm] at (current bounding box.south) {$T_1$};
    \end{tikzpicture}%
    \hspace{2cm}%
    \begin{tikzpicture}
      \path       node [vertex,label=left:$r_C$]           (rA)   {}
      ++(60:1)    node [vertex,label=left:$p_1$]           (p1)   {}
      +(300:1)    node [vertex,label=below:$x_1$]          (x1)   {}
      ++(60:1)    node [vertex,label=left:$p_2$]           (p2)   {}
      +(300:1)    node [vertex,label=below:$x_2$]          (x2)   {}
      +(60:1)     node [vertex,xshift=5mm,label=above:$r$] (r)    {}
      ++(0:2)     node [vertex,label=right:$p_3$]          (p3)   {}
      +(240:1)    node [vertex,label=below:$x_3$]          (x3)   {}
      ++(300:1.5) node [vertex,label=right:$p_{k-1}$]      (pk-1) {}
      +(240:1)    node [vertex,label=below:$x_{k-1}$]      (xk-1) {}
      ++(300:1)   node [vertex,label=right:$p_k$]          (pk)   {}
      +(240:1)    node [vertex,label=below:$x_x$]          (xk)   {}
      ++(300:1)   node [vertex,label=right:$r_D$]          (rB)   {};
      \path (rA) +(270:0.9) node {$T_C$};
      \path (rB) +(270:0.9) node {$T_D$};
      \begin{scope}[on background layer]
        \path [subtree] (rA.center) -- +(240:1.5) -- +(300:1.5) -- cycle;
        \path [subtree] (rB.center) -- +(240:1.5) -- +(300:1.5) -- cycle;
      \end{scope}
      \draw [edge] (rA) -- (p1) -- (x1) (p1) -- (p2) -- (x2) (p2) -- (r) -- (p3) -- (x3)
      (xk-1) -- (pk-1) -- (pk) (xk) -- (pk) -- (rB);
      \draw [edge] (p3) -- (barycentric cs:p3=0.685,pk-1=0.315);
      \draw [edge] (pk-1) -- (barycentric cs:p3=0.315,pk-1=0.685);
      \draw [edge,dotted] (barycentric cs:p3=0.315,pk-1=0.685) -- (barycentric cs:p3=0.685,pk-1=0.315);
      \node [anchor=north,yshift=-5mm] at (current bounding box.south) {$T_2$};
    \end{tikzpicture}}\\[2\bigskipamount]
    \subcaptionbox{}{\begin{tikzpicture}
      \path       node [vertex,label=left:$r_A$]           (rA)   {}
      ++(60:1)    node [vertex,label=left:{$p_1 = q$}]     (p1)   {}
      +(300:1)    node [vertex,label=below:$x_1$]          (x1)   {}
      ++(60:1)    node [vertex,label=left:$p_2$]           (p2)   {}
      +(300:1)    node [vertex,label=below:$x_2$]          (x2)   {}
      +(60:1)     node [vertex,xshift=5mm,label=above:$r$] (r)    {}
      ++(0:2)     node [vertex,label=right:$p_3$]          (p3)   {}
      +(240:1)    node [vertex,label=below:$x_3$]          (x3)   {}
      ++(300:1.5) node [vertex,label=right:$p_{k-1}$]      (pk-1) {}
      +(240:1)    node [vertex,label=below:$x_{k-1}$]      (xk-1) {}
      ++(300:1)   node [vertex,label=right:$p_k$]          (pk)   {}
      +(240:1)    node [vertex,label=below:$x_x$]          (xk)   {}
      ++(300:1)   node [vertex,label=right:$r_B$]          (rB)   {};
      \path (rB) +(270:0.9) node {$T_B$};
      \begin{scope}[on background layer]
        \path [subtree] (rA.center) -- +(240:1.5) coordinate (TAl) -- +(300:1.5) coordinate (TAr) -- cycle;
        \path [subtree] (rB.center) -- +(240:1.5) -- +(300:1.5) -- cycle;
      \end{scope}
      \node [anchor=north] at (barycentric cs:TAl=0.5,TAr=0.5) {$T_A = T'$};
      \draw [edge] (rA) -- (p1) -- (x1) (p1) -- (p2) -- (x2) (p2) -- (r) -- (p3) -- (x3)
      (xk-1) -- (pk-1) -- (pk) (xk) -- (pk) -- (rB);
      \draw [edge] (p3) -- (barycentric cs:p3=0.685,pk-1=0.315);
      \draw [edge] (pk-1) -- (barycentric cs:p3=0.315,pk-1=0.685);
      \draw [edge,dotted] (barycentric cs:p3=0.315,pk-1=0.685) -- (barycentric cs:p3=0.685,pk-1=0.315);
      \node [anchor=north,yshift=-5mm] at (current bounding box.south) {$T_1'$};
    \end{tikzpicture}%
    \hspace{2cm}%
    \begin{tikzpicture}
      \path       node [vertex,label=left:$r_C$]               (rA)   {}
      ++(45:1.2)  node [vertex,label={[yshift=2pt]left:$q$}]   (q)    {}
      +(315:1.2)  node [vertex]                                (rTT)  {}
      ++(45:1.1)  node [vertex,label={[yshift=2pt]left:$p_1$}] (p1)   {}
      +(315:1)    node [vertex,label=below:$x_1$]              (x1)   {}
      ++(45:1.1)  node [vertex,label={[yshift=2pt]left:$p_2$}] (p2)   {}
      +(315:1)    node [vertex,label=below:$x_2$]              (x2)   {}
      +(60:1)     node [vertex,xshift=5mm,label=above:$r$]     (r)    {}
      ++(0:2)     node [vertex,label=right:$p_3$]              (p3)   {}
      +(240:1)    node [vertex,label=below:$x_3$]              (x3)   {}
      ++(300:1.5) node [vertex,label=right:$p_{k-1}$]          (pk-1) {}
      +(240:1)    node [vertex,label=below:$x_{k-1}$]          (xk-1) {}
      ++(300:1)   node [vertex,label=right:$p_k$]              (pk)   {}
      +(240:1)    node [vertex,label=below:$x_x$]              (xk)   {}
      ++(300:1)   node [vertex,label=right:$r_D$]              (rB)   {};
      \path (rA)  +(270:0.9) node {$T_C$};
      \path (rB)  +(270:0.9) node {$T_D$};
      \path (rTT) +(270:0.9) node {$T'$};
      \begin{scope}[on background layer]
        \path [subtree] (rA.center)  -- +(240:1.5) -- +(300:1.5) -- cycle;
        \path [subtree] (rB.center)  -- +(240:1.5) -- +(300:1.5) -- cycle;
        \path [subtree] (rTT.center) -- +(240:1.5) -- +(300:1.5) -- cycle;
      \end{scope}
      \draw [edge] (rA) -- (q) -- (rTT) (q) -- (p1) -- (x1) (p1) -- (p2) -- (x2) (p2) -- (r) -- (p3) -- (x3)
      (xk-1) -- (pk-1) -- (pk) (xk) -- (pk) -- (rB);
      \draw [edge] (p3) -- (barycentric cs:p3=0.685,pk-1=0.315);
      \draw [edge] (pk-1) -- (barycentric cs:p3=0.315,pk-1=0.685);
      \draw [edge,dotted] (barycentric cs:p3=0.315,pk-1=0.685) -- (barycentric cs:p3=0.685,pk-1=0.315);
      \node [anchor=north,yshift=-5mm] at (current bounding box.south) {$T_2'$};
    \end{tikzpicture}}
    \caption{(a) A common chain $\langle x_1, \ldots, x_k \rangle$ of $T_1$ and
    $T_2$ that is pendant in $T_1$: $T_A$ is empty. (b) The two trees $T_1'$ and
    $T_2'$ obtained by attaching a tree $T_A = T'$ to $T_1$ and adding $T'$
    to~$T_C$.}
    \label{fig:attach-subtree}
    \end{figure}

  To complete the proof, we consider the case when at least one of the subtrees
  $T_A, T_B, T_C, T_D$ is empty.  If $T_A = \emptyset$ or $T_C = \emptyset$,
  then we construct a \emph{rooted} tree $T'$ with $t$ leaves.  \nz{(Recall that
  $t = |S|$ is the number of available states.)}  \markj{If $T_A = \emptyset$,
  then we have $p_1 = p_2$.  In this case, we subdivide the edge $(p_2,x_1)$ in
  $T_1$ with a new vertex $q$, and add an edge between $q$ and the root of $T'$.
  This effectively sets $T_A = T'$ and makes $q$ the parent of $x_1$, that is,
  $p_1 = q$ after adding $T'$ to $T_1$.  If $T_A \neq \emptyset$, then we
  subdivide the edge $(p_1,r_A)$ with a new vertex $q$ and again add an edge
  between $q$ and the root of $T'$.  Similarly, in $T_2$ we add an edge between
  the root of $T'$ and a new vertex $q$, where $q$ subdivides the edge
  $(p_2,x_1)$ or $(p_1,r_C)$ depending on whether $S_C = \emptyset$.} This is
  illustrated in \cref{fig:attach-subtree}.  We add a tree $T\dprime$ in a
  similar fashion if $T_B = \emptyset$ or $T_D = \emptyset$.  Let $T_1'$ and
  $T_2'$ be the two trees obtained from $T_1$ and $T_2$ by the addition of~$T'$,
  and possibly~$T\dprime$; let $f'$ be the character on the leaf set of $T_1'$
  and $T_2'$ obtained by setting $f'(v) = f(v)$ for all $v \in X$, giving each
  leaf in $T'$ a different label in $S$, and giving each leaf in $T\dprime$ a
  different label in $S$; and let $F_1'$ and $F_2'$ be the Fitch maps of $f'$ on
  $T_1'$ and $T_2'$, respectively.  \markj{Similar to $p_1, \ldots, p_k$, we use
  $q$ to denote the vertex adjacent to the root of $T'$ in both $T_1'$ and
  $T_2'$.}
  
  Observe that every non-leaf vertex in $T'$ is a union vertex in both $T_1'$
  and~$T_2'$, while $q$ is an intersection vertex in both $T_1'$ and~$T_2'$.
  Moreover, if $T_A = \emptyset$, then $F_1'(q) = \{f(x_1)\} = F_1(x_1)$; if
  $T_A \ne \emptyset$, then $F_1'(q) = F_1(r_A)$.  Similarly, $F_2'(q) =
  F_2(x_1)$ if $T_C = \emptyset$, and $F_2'(q) = F_2(r_C)$ if $T_C \ne
  \emptyset$.  This implies that $F_1(v) = F_1'(v)$ for every vertex $v \in
  T_1$, and $F_2(v) = F_2'(v)$ for every vertex $v \in T_2$.  In particular, the
  addition of $T'$ introduces the non-leaf vertices of $T'$ as union vertices
  into $T_1'$ and $T_2'$ and apart from this, $T_1$ and $T_1'$ have the same
  sets of union vertices, as do $T_2$ and~$T_2'$.  By a similar argument, the
  addition of $T\dprime$ if $T_B = \emptyset$ or $T_D = \emptyset$ introduces
  the same number of union vertices into both $T_1'$ and~$T_2'$.
  
  This implies that
  \begin{equation*}
    l_f(T_2) - l_f(T_1) = l_{f'}(T_2') - l_{f'}(T_1').
  \end{equation*}
  
  Finally observe that the sets $S_A$, $S_B$, $S_C$, and $S_D$, \markj{and thus
  $\delta_{AC}, \delta_{BD}$,} are the same for $T_1$ and $T_2$ as for $T_1'$
  and $T_2'$.  Indeed, if $T_A \ne \emptyset$, then $S_A = F_1(r_A) = F_1'(q)$.
  If $T_A = \emptyset$, then we define $S_A = S$ for $T_1$.  In~$T_1'$, $S_A =
  S$ because $T_A = T'$ in $T_1'$ and all internal vertices in $T'$ are union
  vertices.  In both cases, $q$ plays the role of $r_A$ in~$T_1'$.  Analogous
  arguments show that $S_B$, $S_C$, and $S_D$ are the same for $T_1$ and $T_2$
  as for $T_1'$ and $T_2'$.

  By applying the case when $T_A, T_B, T_C, T_D$
  are all non-empty to $T_1'$ and $T_2'$, we conclude that
  \begin{equation*}
    l_f(T_2) - l_f(T_1) = l_f(T_2') - l_f(T_1') \le
    l_{f_C}(T_C) + l_{f_D}(T_D) - l_{f_A}(T_A) - l_{f_B}(T_B) + \delta_{AC} + \delta_{BD},
  \end{equation*}
  that is, once again,
  \begin{equation*}
    u_2 - u_1 \le \delta_{AC} + \delta_{BD}.
  \end{equation*}
\end{proof}

We are ready to prove that chain reduction is safe now:

\begin{lem}
  \label{lem:ChainReduction}
  Let $T_1$ and $T_2$ be two 
  trees on $X$, let $\langle x_1,
  \ldots, x_k \rangle$ be a common chain of $T_1$ and $T_2$ of length $k \geq
  5$, and let $T_1'$ and $T_2'$ be the two trees obtained by removing the leaves
  $x_3, \ldots, x_{k-2}$ from both $T_1$ and $T_2$ and suppressing their
  parents.  Then $\dmp^t(T_1, T_2) = \dmp^t(T'_1, T'_2)$ for all $t \in \trange$.
\end{lem}

\begin{proof}
  The proof is based on the proof by Kelk et al.\
  \cite{kelkReductionRulesMaximum2016} but presents the argument much more
  succinctly, and obviously makes adjustments to ensure that the proof is
  correct for $t$-state characters.
  
  By \cref{lem:subset-restriction}, we have that $\dmp^t(T_1', T_2') \le
  \dmp^t(T_1, T_2)$. Therefore, it suffices to show that $\dmp^t(T_1', T_2') \ge
  \dmp^t(T_1, T_2)$.  Let $f$ be an optimal character for $(T_1, T_2)$, and
  assume that $\dmp^t(T_1, T_2) = \PS{f}{T_2} - \PS{f}{T_1}$.  We construct a
  $t$-state character $f'$ on the leaf set $X \setminus \{x_3, \ldots,
  x_{k-2}\}$ of $T_1'$ and $T_2'$ such that $l_{f'}(T_2') - l_{f'}(T_1') \ge
  l_f(T_2) - l_f(T_1) = \dmp^t(T_1, T_2)$.  Since $\dmp^t(T_1', T_2') \ge
  l_{f'}(T_2') - l_{f'}(T_1')$, this proves the claim.  We define $f'$ as
  \begin{equation*}
    f'(v) = \begin{cases}
      a    & \text{if } v \in \{x_1, x_2\}\\
      b    & \text{if } v \in \{x_{k-1}, x_k\}\\
      f(v) & \text{otherwise,}
    \end{cases}
  \end{equation*}
  where $a, b \in S$ are appropriate states chosen as discussed below.

  For this character, we use $F_1'$ and $F_2'$ to denote the Fitch maps it
  defines on $T_1'$ and~$T_2'$.  Note that the choice of $f'$ ensures that $p_2$
  and $p_{k-1}$ are intersection vertices in both $T_1'$ and $T_2'$ and that
  $F_1'(p_2) = F_2'(p_2) = \{a\}$ and $F_1'(p_{k-1}) = F_2'(p_{k-1}) = \{b\}$.
  In turn, the latter implies that $r$ is a union vertex in $T_1'$ if and only
  if it is a union vertex in~$T_2'$.  Thus,
  \begin{equation*}
    l_{f'}(T_2') - f_{f'}(T_1') = l_{f_C}(T_C) + l_{f_D}(T_D) - l_{f_A}(T_A) - l_{f_B}(T_B) + \chi_1 + \chi_k,
  \end{equation*}
  where
  \begin{equation*}
    \chi_i = \begin{cases}
      -1 & \text{if $p_i$ is a union vertex in $T_1'$ but not in $T_2'$}\\
      1  & \text{if $p_i$ is a union vertex in $T_2'$ but not in $T_1'$}\\
      0  & \text{otherwise,}
    \end{cases}
  \end{equation*}
  for $i \in \{1, k\}$.
  
  By \cref{eq3.2} and \cref{lemma:lcra}, we have that
  \begin{equation*}
    \markj{l_f(T_2) - l_f(T_1) = \dmp^t(T_1, T_2) \le l_{f_C}(T_C) + l_{f_D}(T_D) - l_{f_A}(T_A) - l_{f_B}(T_B) + \delta_{AC} + \delta_{BD}.}
  \end{equation*}
  Thus, to prove that $l_{f'}(T_2') - l_{f'}(T_1') \ge l_f(T_2) - l_f(T_1)$, it
  suffices to prove that we can choose $f'$ so that $\chi_1 \ge
  \delta_{AC}$ and $\chi_k \ge \delta_{BD}$.

  We prove that we can choose $a$ so that $\chi_1 \ge \delta_{AC}$.  An
  analogous argument shows that we can choose $b$ so that $\chi_k \ge
  \delta_{BD}$.

  We choose $a \in S_A$.  If $T_A \ne \emptyset$, this ensures that $p_1$ is an
  intersection vertex.  If $T_A = \emptyset$, then $p_1 = p_2$ and $f'(x_1) =
  f'(x_2) = a$, so \emph{any} choice of $a$ ensures that $p_1$ is an
  intersection vertex.  Thus, $\chi_1 \ge 0$.  In particular, $\chi_1 \ge
  \delta_{AC}$ if $\delta_{AC} = 0$.

  If $\delta_{AC} = 1$, then $S_A \not\subseteq S_C$.  Thus, we can choose $a
  \in S_A \setminus S_C$.  This ensures not only that $p_1$ is an intersection
  vertex in $T_1'$ but also that it is a union vertex in~$T_2'$.  (In
  particular, $S_A \not\subseteq S_C$ implies that $S_C \ne S$, so $T_C \ne
  \emptyset$.)  Thus, $\chi_1 = 1 = \delta_{AC}$ in this case.
\end{proof}

\section{\boldmath A Lower Bound on $\dmp^t$}

\label{sec:kernel-size}

By \cref{thm:reduction-rules}, parsimony distance has a kernel of size linear in
the TBR distance between the two trees.  In this section, we bound the size of
this kernel as a function of the parsimony distance itself.  Specifically, we
prove the following result:

\begin{thm}
  \label{thm:lower-bound}
  Any two trees $T_1$ and $T_2$ on $X$ satisfy $\dtbr(T_1, T_2) \le 54k(\lg |X| +
  1)$, where $k = \dmp^t(T_1, T_2)$, for any $t \in \trange$.
\end{thm}

Together with \cref{thm:reduction-rules}, this implies the following corollary:

\begin{cor}
  \label{cor:kernel-size}
  There exists a set of reduction rules for $\dmp^t$ such that a fully reduced
  yes-instance $(T_1, T_2, t, k)$ with $t \in \trange$ and $k \ge 1$ consists of
  a pair of trees on $T_1$ and $T_2$ on $X$ with $|X| \le 1{,}484k (\lg k + 11) \in
  O(k \lg k)$.
\end{cor}

\begin{proof}
  Let $(T_1, T_2, t, k)$ be a fully reduced yes-instance, let $n = |X|$, and
  let $k' = \dtbr(T_1, T_2)$. By \cref{thm:reduction-rules},
  we have
  \begin{align*}
    n  &\le 20k'\\
  \intertext{and, \markj{by \cref{thm:lower-bound},}}
    k' &\le 54k(\lg n + 1).
  \end{align*}
  This gives
  \begin{equation}\label{eq:nbound}
    n \le 1{,}080k(\lg n + 1).
  \end{equation}
  For notational convenience, let $c = 11 \cdot 1{,}484 = 16{,}324$.
  If $n \le c$, then the bound on $n$ claimed in the corollary holds.
  So assume that $n > c$.
  Then $\lg n > \lg c$, so $\lg n + 1 < \bigl(1 + \frac{1}{\lg c}\bigr)\lg
  n$ and
  \begin{equation*}
    n < 1{,}080\left(1 + \frac{1}{\lg c}\right)k\lg n \le 1{,}158 k \lg n.
  \end{equation*}
  Since $n > c \ge 8$, we also have $\lg n \le n^{\frac{\lg\lg c}{\lg c}}$
  \markj{(indeed, for $n \ge 8$, we have $\lg n = n^{\frac{\lg\lg n}{\lg n}}$, and $\frac{\lg\lg n}{\lg n}$ is a decreasing function. Therefore,}
  \begin{align*}
    n &\le 1{,}158kn^{\frac{\lg\lg c}{\lg c}},\\
    n^{\frac{\lg c - \lg\lg c}{\lg c}} &\le 1{,}158k,\\
    n &\le (1{,}158k)^{\frac{\lg c}{\lg c - \lg\lg c}}.
  \end{align*}
  This implies that
  \begin{equation*}
    \lg n \le \frac{\lg c}{\lg c - \lg\lg c} \cdot (\lg k + \lg 1{,}158),
  \end{equation*}
  so \markj{by \cref{eq:nbound},}
  \begin{equation*}
    n \le 1{,}080k\left(\frac{\lg c}{\lg c - \lg\lg c} \cdot (\lg k + \lg 1{,}158) + 1\right)
    \le 1{,}484k(\lg k + 11).
  \end{equation*}
\end{proof}

\markj{It remains to prove \cref{thm:lower-bound}. As stated in the
introduction, the key}
is to show that $\dmp^t(T_1, T_2)$ is large if $T_1$ and $T_2$ have a large
number of ``leg-disjoint'' incompatible quartets.  This is similar to the lower
bound on $\dmp^\infty(T_1, T_2)$ by Jones, Kelk, and Stougie
\cite{jonesMaximumParsimonyDistance2021}, where it was shown that
$\dmp^\infty(T_1, T_2)$ is large if $T_1$ and $T_2$ have a large number of
\emph{disjoint} incompatible quartets.  Leg-disjointness is a much weaker
condition.  We define the concept of leg-disjoint incompatible quartets in
\cref{sec:leg-disjoint-quartets}, and show that they provide a lower bound on
$\dmp^t(T_1, T_2)$, for any $t \in \trange$.  To prove \cref{thm:lower-bound},
we then show, in \cref{sec:finding-quartets}, how to find a set of at least
$\frac{\dtbr(T_1, T_2)}{2(\lg |X| + 1)}$ leg-disjoint incompatible quartets.

\subsection{Leg-Disjoint Quartets}

\label{sec:leg-disjoint-quartets}

Given a tree $T$ on $X$, we call two quartets $q_1, q_2 \subseteq X$ \emph{fully
$T$-disjoint} if $T(q_1)$ and $T(q_2)$ are disjoint.  Given a quartet $q =
\{a,b,c,d\}$ such that $T|_q = ab|cd$, we call the paths from $a$ to $b$ and
from $c$ to $d$ in $T$ the \emph{legs} of~$q$.  The path composed of all edges
in $T(q)$ not included in the legs of $q$ is the \emph{backbone} of~$q$.  The
endpoints of the backbone are the \emph{joints} of~$q$.  We call two quartets
\emph{$T$-leg-disjoint} if their legs in $T$ are disjoint.
Note that this implies that the quartets are themselves disjoint \markj{subsets of $X$.}

The main result in this section proves that $\dmp^t(T_1, T_2)$ is large if there
exists a large set $Q$ of pairwise $T_1$-leg-disjoint incompatible quartets of
$T_1$ and~$T_2$.  In the remainder of this section, we refer to the quartets
in $Q$ simply as leg-disjoint, omitting the explicit reference to the tree $T_1$
in which their legs are disjoint.

\begin{prop}
  \label{prop:leg-disjoint-quartets}
  Let $Q$ be a set of pairwise leg-disjoint incompatible quartets of two trees
  $T_1$ and $T_2$ on $X$.  Then $\dmp^t(T_1, T_2) \ge \frac{|Q|}{27}$, for all
  $t \in \trange$.
\end{prop}

Note that \cref{prop:leg-disjoint-quartets} does not impose \emph{any}
constraints on the manner in which the quartets in $Q$ interact in~$T_2$, nor
does it require their backbones in $T_1$ to be disjoint from each other or from
the legs of other quartets in~$Q$.  Contrast this with the definition of
disjoint quartets used by Jones, Kelk, and Stougie
\cite{jonesMaximumParsimonyDistance2021}, which considers two quartets $q_1$ and
$q_2$ to be disjoint if they are both fully $T_1$-disjoint and fully
$T_2$-disjoint.

\begin{proof}
  To simplify the proof, we may assume w.l.o.g.\ that every leaf in $X$ is
  part of a quartet in $Q$. Indeed, if this is not the case, then let $Y \subset
  X$ be the set of leaves that belong to quartets in~$Q$.  By
  \cref{lem:subset-restriction}, we have $\dmp^2(T_1, T_2) \ge \dmp^2(T_1|_Y,
  T_2|_Y)$, and $Q$ is also a set of pairwise leg-disjoint incompatible quartets
  of $T_1|_Y$ and~$T_2|_Y$. \nz{Therefore,} we may replace $T_1$ and $T_2$ with $T_1|_Y$
  and $T_2|_Y$ in what follows.
  
  To prove the proposition, we construct a subset $Q' \subseteq Q$ such that
  $|Q'| \ge \frac{|Q|}{9}$ and $\dmp^t(T_1, T_2) \ge \frac{|Q'|}{3}$.  Thus,
  $\dmp^t(T_1, T_2) \ge \frac{|Q|}{27}$, as claimed.

  To describe the construction of this subset $Q' \subseteq Q$, we need some
  notation.  We use $X'$ to refer to the set of leaves of the quartets in $Q'$:
  $X' = \bigcup_{q \in Q'} q$.  Let $q = \{a, b, c, d\}$ be an incompatible
  quartet and assume that $T_1|_q = ab|cd$.  For a labelling $\bar f :
  V(T_2(X')) \rightarrow S$ of some subtree $T_2(X')$ of $T_2$ with $q \subseteq
  X'$, we define
  \begin{equation*}
    \beta_{\bar f}(q) = |\{(x,y) \in \{(a,b), (c,d)\} \mid \bar f(x) \ne \bar f(y)\}|.
  \end{equation*}
  In words, $\beta_{\bar f}(q)$ is the number of legs of $q$ in $T_1$ whose
  endpoints are assigned different states by~$\bar f$.  Furthermore, for any
  subset $Q\dprime \subseteq Q$ such that every quartet $q \in Q\dprime$
  satisfies $q \subseteq X'$, we define
  \begin{equation*}
    \beta_{\bar f}(Q\dprime) = \sum_{q \in Q\dprime} \beta_{\bar f}(q).
  \end{equation*}

  Since the quartets in $Q' \subseteq Q$ are pairwise leg-disjoint in $T_1$, any
  character $f$ on $X'$ satisfies $l_f(T_1(X')) \ge \beta_{\bar f}(Q')$, where
  $\bar f$ is an arbitrary extension of $f$ to $T_2(X')$.  Indeed, every leg of
  a quartet $q \in Q'$ that contributes to $\beta_{\bar f}(q)$ must include a
  mutation edge and thus increases $l_f(T_1(X'))$ by 1 because the quartets in
  $Q'$ are pairwise leg-disjoint.  We also have $\Delta_{\bar f}(T_2(X')) \ge
  l_{f}(T_2(X'))$.  Thus, it suffices to construct a subset $Q' \subseteq Q$ and
  an extension $\bar f$ of a 2-state character $f : X' \rightarrow S$ to the
  vertices of $T_2(X')$ such that $|Q'| \ge \frac{|Q|}{9}$ and $\beta_{\bar
  f}(Q') - \Delta_{\bar f}(T_2(X')) \ge \frac{|Q'|}{3}$.  Indeed, this implies
  that $\dmp^2(T_1(X'), T_2(X')) \ge l_{f}(T_1(X')) - l_{f}(T_2(X')) \ge
  \beta_{\bar f}(Q') - \Delta_{\bar f}(T_2(X')) \ge \markj{\frac{|Q'|}{3}
  \ge}\frac{|Q|}{27}$.  By \cref{lem:subset-restriction}, we have $\dmp^t(T_1,
  T_2) \ge \dmp^2(T_1, T_2) \ge \dmp^2(T_1(X'), T_2(X'))$, so $\dmp^t(T_1, T_2)
  \ge \frac{|Q|}{27}$.

  We assume that $S = \{\text{red}, \text{blue}\}$ from here on.  Accordingly,
  we call the states in $S$ \emph{colours} and refer to $\bar f: V(T_2(X'))
  \rightarrow S$ as a \emph{colouring} of $T_2(X')$.

  We construct the desired subset $Q' \subseteq Q$ and colouring $\bar f$ of
  $T_2(X')$ in two phases, maintaining the invariant that $\beta_{\bar f}(Q') -
  \Delta_{\bar f}(T_2(X')) \ge \frac{|Q'|}{3}$.  We prove that once we are
  unable to find more quartets to add to~$Q'$ in the second phase, we have $|Q'|
  \ge \frac{|Q|}{9}$.  Thus, the set $Q'$ and the colouring $\bar f$ obtained
  after the second phase have the desired properties.

  \paragraph{\boldmath Phase 1: Select a maximal subset of pairwise
  fully $T_2$-disjoint quartets.}

  We select a maximal subset $Q' \subseteq Q$ of quartets that are pairwise
  fully $T_2$-disjoint.  The vertices of $T_2(X')$ can easily be coloured so
  that $\beta_{\bar f}(Q') = 2|Q'|$ and $\Delta_{\bar f}(T_2(X')) = |Q'|$:
  
  Consider the forest $F$ obtained from $T_2(X')$ by deleting one edge $e_q$
  from the backbone in $T_2$ of each quartet $q \in Q'$.  Let $T'$ be the tree
  obtained from $T_2(X')$ by contracting every connected component of $F$ into a
  single vertex.  Since $T'$ is a tree, it is bipartite and thus can be
  $2$-coloured.  Let $f' : V(T') \rightarrow \{\text{red}, \text{blue}\}$
  be such a $2$-colouring of~$T'$.  Then we choose $\bar f$ so that it colours
  every vertex in the connected component of $F$ represented by $v$ with the
  colour $f'(v)$, for every vertex $v \in V(T')$.

  Since $F$ has $|Q'| + 1$ connected components, we have $\Delta_{\bar
  f}(T_2(X')) = |Q'|$.
  
  Next consider any quartet $q \in Q'$ and assume w.l.o.g.\ that $T_1|_q =
  ab|cd$ and $T_2|_q = ac|bd$.  Then $a$ and $c$ belong to the same connected
  component $C_1$ of $F$, $b$ and $d$ belong to the same connected component
  $C_2$ of $F$, $C_1 \ne C_2$, and the two vertices $v_1$ and $v_2$ in $T'$
  representing $C_1$ and $C_2$ are adjacent.  Indeed, if $a$ and $c$ belonged to
  different connected components of~$F$, $b$ and $d$ belonged to different
  connected components of $F$ or $v_1$ and $v_2$ were not adjacent in $T'$, then
  $T_2(q)$ would contain an edge $e_{q'}$ in the backbone of another quartet $q'
  \in Q'$, a contradiction because the quartets in $Q'$ are fully
  $T_2$-disjoint.  The fact that $C_1$ and $C_2$ are different connected
  components follows because deleting the edge $e_q$ separates $a$ and $c$ from
  $b$ and $d$ in~$T_2$.

  Since $a, c \in C_1$, $b, d \in C_2$, and $v_1$ and $v_2$ are adjacent, we
  have $\bar f(a) \ne \bar f(b)$ and $\bar f(c) \ne \bar f(d)$.  Thus,
  $\beta_{\bar f}(q) = 2$.  Since this is true for every quartet $q \in Q'$, we
  have $\beta_{\bar f}(Q') = 2|Q'|$.  This shows that $\beta_{\bar f}(Q') -
  \Delta_{\bar f}(T_2(X')) = |Q'| \ge \frac{|Q'|}{3}$.

  \paragraph{\boldmath Phase 2: Greedily add quartets to $Q'$.}

  Let $U = Q \setminus Q'$ be the set of uncoloured quartets.  We add quartets
  from $U$ to $Q'$ one, two or three quartets at a time.  For each added group
  of quartets, we extend the colouring $\bar f$ to $T_2(X')$ and modify it in a
  manner that ensures that $\beta_{\bar f}(Q')  - \Delta_{\bar f}(T_2(X'))$
  increases by at least~1.  Thus, the inequality $\beta_{\bar f}(Q')  -
  \Delta_{\bar f}(T_2(X')) \ge \frac{|Q'|}{3}$ is maintained by each addition.
  
  To move quartets from $U$ to $Q'$, we consider several cases, choosing the
  first case that applies:

  \paragraph{Case 1: Quartets with a ``good'' parsimonious extension.}

  If there exists a quartet $q \in U$ such that the parsimonious extension $\bar
  f'$ of $\bar f$ to $T_2(X' \cup q)$ satisfies $\beta_{\bar f'}(q) > 0$, then
  we add $q$ to $Q'$ and set $\bar f = \bar f'$.  See \cref{fig:ldiq-case-1}.
  This increases $\beta_{\bar f}(Q')$ by at least 1 and leaves $\Delta_{\bar
  f}(T_2(X'))$ unchanged \markj{(note that since we add $q$ to $Q'$, $T_2(X')$ now includes the leaves of $q$)}.  Thus, $\beta_{\bar f}(Q') - \Delta_{\bar f}(T_2(X'))$
  increases by at least~1.

  \bigskip

  The remaining cases assume that Case~1 is not applicable.  Thus, the
  parsimonious extension $\bar f'$ of $\bar f$ to $T_2(X' \cup q)$ satisfies
  $\beta_{\bar f'}(q) = 0$ for every quartet $q \in U$.  Consider the pendant
  subtrees of $T_2(X')$ in~$T_2$. Since we initialized $Q'$ to be a maximal
  subset of quartets that are pairwise fully $T_2$-disjoint, there is no quartet
  in $U$ that has all its leaves in one of these pendant subtrees.

  \paragraph{Case 2: Quartets with at least two leaves in the same pendant
    subtree.}

  Suppose that there exists a quartet $q \in U$ that has at least two of its
  leaves in the same pendant subtree $T'$ of $T_2(X')$.  The parsimonious
  extension $\bar f'$ of $\bar f$ to $T_2(X' \cup q)$ colours all leaves of $q$
  in $T'$ the same colour.  Assume that $T_1|_q = ab|cd$ and $T_2|_q = ac|bd$.
  Since $T'$ contains at least two leaves of $q$, we can assume w.l.o.g.\ that
  $a, c \in T'$, and that $\bar f'$ colours $a$ and $c$ red.  Since $\beta_{\bar
  f'}(q) = 0$, this implies that $\bar f'$ also colours $b$ and $d$ red.  We add
  $q$ to $Q'$, set $\bar f = \bar f'$, and then change the colour of $a$ and $c$
  to blue and colour all vertices on the path from $a$ to $c$ in $T_2$ blue.
  See \cref{fig:ldiq-case-2}.  This ensures that $\beta_{\bar f}(q) = 2$, so
  $\beta_{\bar f}(Q')$ increases by~2.  $\Delta_{\bar f}(T_2(X'))$ is easily
  verified to increase by~1 \markj{(since $(a,c)$ must be a cherry in
  $T_2|_{X'}$)}.  Thus, $\beta_{\bar f}(Q') - \Delta_{\bar f}(T_2(X'))$
  increases by~1.

  \bigskip

  If neither Case~1 nor Case~2 applies, then every pendant subtree of
  $T_2(X')$ contains at most one leaf from each quartet in $U$.

  \paragraph{Case 3: A pendant subtree with leaves from more than one quartet.}

  If there exists a pendant subtree $T'$ of $T_2(X')$ that contains leaves from
  at least two quartets in $U$, then pick two such quartets $q_1$ and $q_2$ and
  let $\bar f'$ be the parsimonious extension of $\bar f$ to $T_2(X' \cup q_1
  \cup q_2)$.  Assume that $\bar f'$ colours the leaves of $q_1 \cup q_2$ in
  $T'$ red.  We add $q_1$ and $q_2$ to $Q'$ and set $\bar f = \bar f'$.  Then we
  change the colour of every vertex in $T_2(X')$ that belongs to $T'$ to blue.
  (Since we added $q_1$ and $q_2$ to $Q'$, $T_2(X')$ now includes vertices
  in~$T'$.)  See \cref{fig:ldiq-case-3}.  Since both $q_1$ and $q_2$ have a
  single leaf in $T'$, this ensures that $\beta_{\bar f}(Q') $ increases by~2,
  whereas $\Delta_{\bar f}(T_2(X'))$ increases by~1.  Thus, $\beta_{\bar f}(Q')
  - \Delta_{\bar f}(T_2(X'))$ increases by~1.

  \begin{figure}[p]
    \centering
    \subcaptionbox{Case~1\label{fig:ldiq-case-1}}{\begin{tikzpicture}
      \path         node [vertex,red node]                (e) {}
      +(150:1)      node [vertex,red node]                (f) {}
      +(210:1)      node [vertex,red node]                (g) {}
      ++(0:1)       node [vertex,red node]                (h) {}
      ++(0:1)       node [vertex,blue node]               (i) {}
      ++(60:1)      node [vertex,blue node]               (j) {}
      +(90:1)       node [vertex,blue node]               (k) {}
      ++(30:1)      node [vertex,blue node]               (l) {}
      +(30:1)       node [vertex,blue node]               (m) {}
      +(300:1)      node [vertex,blue node,label=300:$b$] (b) {}
      (i) ++(330:1) node [vertex,blue node]               (n) {}
      ++(0:1)       node [vertex,blue node]               (o) {}
      +(30:1)       node [vertex,blue node]               (p) {}
      +(330:1)      node [vertex,blue node]               (q) {}
      (n) ++(300:1) node [vertex,red node]                (r) {}
      +(300:1)      node [vertex,red node]                (s) {}
      +(210:1)      node [vertex,red node,label=190:$d$]  (d) {}
      (h) ++(270:1) node [vertex,red node]                (t) {}
      +(240:1)      node [vertex,red node,label=260:$a$]  (a) {}
      ++(300:1)     node [vertex,red node,label=280:$c$]  (c) {};
      \path [bold edge] (h) -- (i) (n) -- (r);
      \path [bold edge,red edge] (g) -- (e) -- (f) (e) -- (h) (r) -- (s);
      \path [bold edge,blue edge] (i) -- (j) -- (k)
      (j) -- (l) -- (m) (i) -- (n) -- (o) -- (p) (o) -- (q);
      \path [thin edge,blue edge] (l) -- (b);
      \path [thin edge,red edge] (r) -- (d) (a) -- (t) -- (c) (t) -- (h);
    \end{tikzpicture}}%
    \hspace{2cm}%
    \subcaptionbox{Case~2\label{fig:ldiq-case-2}}{\begin{tikzpicture}
      \path         node [vertex,red node]                (e) {}
      +(150:1)      node [vertex,red node]                (f) {}
      +(210:1)      node [vertex,red node]                (g) {}
      ++(0:1)       node [vertex,red node]                (h) {}
      ++(0:1)       node [vertex,red node]                (i) {}
      ++(45:1)      node [vertex,red node]                (j) {}
      +(75:1)       node [vertex,red node]                (k) {}
      ++(15:1)      node [vertex,red node]                (l) {}
      +(15:1)       node [vertex,red node]                (m) {}
      +(285:1)      node [vertex,red node,label=275:$b$]  (b) {}
      (i) ++(315:1) node [vertex,blue node]               (n) {}
      ++(345:1)     node [vertex,blue node]               (o) {}
      +(15:1)       node [vertex,blue node]               (p) {}
      +(315:1)      node [vertex,blue node]               (q) {}
      (n) ++(285:1) node [vertex,blue node]               (r) {}
      (h) ++(270:1) node [vertex,red node]                (s) {}
      +(300:1)      node [vertex,red node,label=290:$d$]  (d) {}
      ++(240:1)     node [vertex,blue node]               (t) {}
      +(240:1)      node [vertex,blue node,label=260:$a$] (a) {}
      +(300:1)      node [vertex,blue node,label=280:$c$] (c) {};
      \path [bold edge] (n) -- (i);
      \path [bold edge,red edge] (g) -- (e) -- (f) (e) -- (h) -- (i) -- (j) -- (k)
      (j) -- (l) -- (m);
      \path [bold edge,blue edge] (n) -- (r)
      (n) -- (o) -- (p) (o) -- (q);
      \path [thin edge,red edge] (l) -- (b) (h) -- (s) -- (d);
      \path [thin edge,blue edge] (a) -- (t) -- (c);
      \path [thin edge] (s) -- (t);
    \end{tikzpicture}}\\[2\bigskipamount]
    \subcaptionbox{Case~3\label{fig:ldiq-case-3}}{\begin{tikzpicture}
      \path         node [vertex,red node]                                (e)  {}
      +(90:1)       node [vertex,red node,label=90:$c_1$]                 (c1) {}
      +(180:1)      node [vertex,red node]                                (g)  {}
      ++(0:1)       node [vertex,red node]                                (h)  {}
      ++(0:1)       node [vertex,red node]                                (i)  {}
      ++(60:1)      node [vertex,red node]                                (j)  {}
      +(150:1)      node [vertex,red node,label=170:$d_2$]                (d2) {}
      ++(60:1)      node [vertex,red node]                                (l)  {}
      +(60:1)       node [vertex,red node]                                (m)  {}
      +(330:1)      node [vertex,red node,label=350:$b_2$]                (b2) {}
      (i) ++(330:1) node [vertex,red node]                                (n)  {}
      ++(10:1)      node [vertex,blue node]                               (o)  {}
      +(40:1)       node [vertex,blue node]                               (p)  {}
      +(340:1)      node [vertex,blue node]                               (q)  {}
      (n) ++(315:1) node [vertex,red node]                                (r)  {}
      +(310:1)      node [vertex,red node]                                (s)  {}
      +(220:1)      node [vertex,red node,label=220:$c_2$]                (c2) {}
      (h) ++(270:1) node [vertex,blue node]                               (t)  {}
      +(240:1)      node [vertex,blue node,label={[xshift=1pt]267:$a_1$}] (a1) {}
      +(300:1)      node [vertex,blue node,label=290:$a_2$]               (a2) {}
      (g) ++(135:1) node [vertex,red node]                                (f)  {}
      +(135:1)      node [vertex,red node]                                (u)  {}
      +(225:1)      node [vertex,red node,label=225:$b_1$]                (b1) {}
      (g) ++(225:1) node [vertex,red node]                                (v)  {}
      +(225:1)      node [vertex,red node]                                (w)  {}
      +(315:1)      node [vertex,red node,label=315:$d_1$]                (d1) {};
      \path [bold edge] (n) -- (o);
      \path [bold edge,blue edge] (p) -- (o) -- (q);
      \path [bold edge,red edge] (m) -- (l) -- (j) -- (i) -- (n) -- (r) -- (s)
      (u) -- (f) -- (g) -- (v) -- (w) (g) -- (e) -- (h) -- (i);
      \path [thin edge] (h) -- (t);
      \path [thin edge,blue edge] (a1) -- (t) -- (a2);
      \path [thin edge,red edge] (b1) -- (f) (d1) -- (v) (c2) -- (r)
      (b2) -- (l) (d2) -- (j) (c1) -- (e);
    \end{tikzpicture}}
    \caption{The updated colouring in Cases~1--3 of the proof of
    \cref{prop:leg-disjoint-quartets}.  Only $T_2(X')$ and the leaves of $q$ are
    shown.  Bold solid edges are in $T_2(X')$.  Thin dashed edges are the new
    edges added to $T_2(X' \cup q)$ or, in Case~3, to $T_2(X' \cup q_1 \cup
    q_2)$.}
  \end{figure}
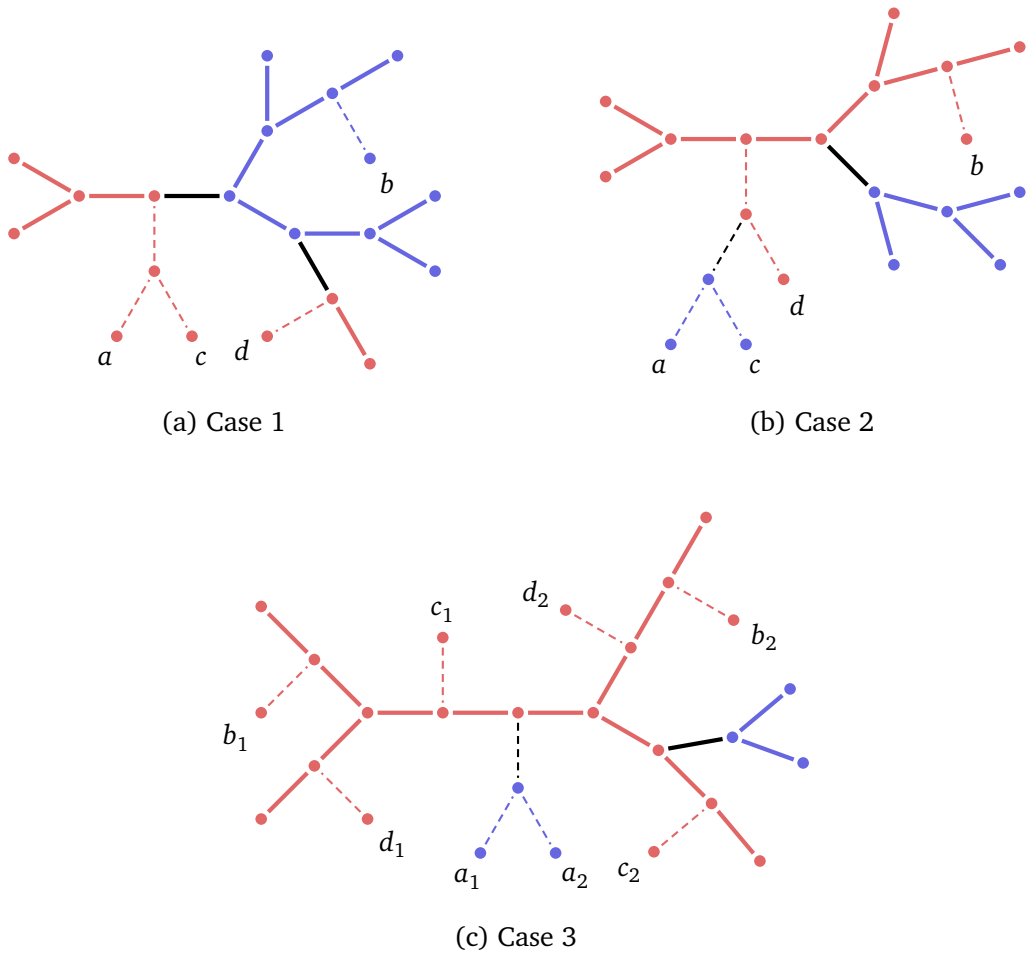
  \begin{figure}[p]
    \centering
    \begin{tikzpicture}
      \path        node [vertex]                                          (e) {}
      ++(0:1)      node [vertex,red node]                                 (f) {}
      +(270:1)     node [vertex,red node,label=below:{\vphantom{$b$}$a$}] (a) {}
      ++(0:1)      node [vertex,red node]                                 (g) {}
      +(270:1)     node [vertex,red node,label=below:$b$]                 (b) {}
      ++(0:1)      node [vertex,red node]                                 (h) {}
      ++(300:1)    node [vertex]                                          (i) {}
      +(240:1)     node [vertex]                                          (j) {}
      +(0:1)       node [vertex]                                          (k) {}
      (h) ++(60:1) node [vertex,red node]                                 (l) {}
      +(120:1)     node [vertex]                                          (m) {}
      ++(0:1)      node [vertex,red node]                                 (n) {}
      +(90:1)      node [vertex,red node,label=above:$c$]                 (c) {}
      ++(0:1)      node [vertex,red node]                                 (o) {}
      +(90:1)      node [vertex,red node,label=above:$d$]                 (d) {}
      +(0:1)       node [vertex]                                          (p) {};
      \begin{scope}[on background layer]
        \path [subtree] (e.center) -- +(150:1) -- +(210:1) -- cycle;
        \path [subtree] (p.center) -- +(30:1) -- +(330:1) -- cycle;
        \path [subtree] (m.center) -- +(90:1) -- +(150:1) -- cycle;
        \path [subtree] (k.center) -- +(30:1) -- +(330:1) -- cycle;
        \path [subtree] (j.center) -- +(210:1) -- +(270:1) -- cycle;
      \end{scope}
      \path [thin edge,solid] (e) -- (f) (j) -- (i) -- (k) (h) -- (i)
      (l) -- (m) (o) -- (p);
      \path [bold edge,red edge] (a) -- (f) -- (g) -- (b)
      (c) -- (n) -- (o) -- (d) (g) -- (h) -- (l) -- (n);
    \end{tikzpicture}
    \caption{A quartet $q$ with $T_1(q) = ab|cd$ and with $a$ and $b$ \markj{adjacent} to
    one side of $T_2$ and $c$ and $d$ \markj{adjacent} to another side of $T_2$ cannot be
    incompatible.}
    \label{fig:ldiq-case-4}
  \end{figure}
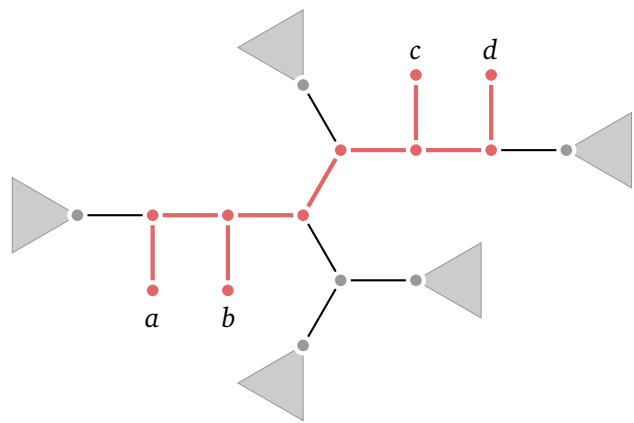

  \paragraph{Case 4: All pendant subtrees are singletons.}

  If we reach this case, then Cases 1--3 do not apply.  Thus, every pendant
  subtree of $T_2(X')$ contains at most one leaf that belongs to a quartet
  in~$U$.  Since we assumed that every leaf in $X$ belongs to some quartet in
  $Q$, this implies that every pendant subtree of $T_2(X')$ consists of a single
  leaf (and this leaf belongs to some quartet in $U$).

  Let a \emph{side} of $T_2(X')$ be a maximal path in $T_2(X')$ whose internal
  vertices have degree $2$ in $T_2(X')$.  Then every leaf \markj{$l$ of a
  quartet in $U$ is adjacent to some side of $T_2(X')$ (that is, $l$ is adjacent
  to an internal vertex of that side).}  Moreover, for any quartet $q \in U$
  with $T_1|_q = ab|cd$, either all leaves of $q$ are adjacent to the same side,
  or at least one of the pairs $\{a,b\}, \{c,d\}$ has its elements
  \markj{adjacent} to two different sides.  If this were not the case, then
  $a,b$ would be \markj{adjacent} to one side and $c,d$ would be
  \markj{adjacent} to another.  This would imply that $T_2|q = ab|cd$,
  contradicting that $q$ is an incompatible quartet.  See
  \cref{fig:ldiq-case-4}.  We can assume from here on that all internal vertices
  of a side of $T_2(X')$ have the same colour.  If not, we can change $\bar f$
  so that this is true, without increasing $\Delta_{\bar f}(T_2(X'))$.

  \paragraph{Case 4.1: A side with \markj{adjacent} leaves from at least three quartets.}

  Suppose that there exists a side $P$ that has \markj{adjacent} leaves from at least
  three quartets $q_1, q_2, q_3 \in U$ and for each $i \in \{1,2,3\}$, $P$~has
  $a_i$ but not $b_i$ as an \markj{adjacent} vertex, where $T_1|{q_i} = a_ib_i|c_id_i$.
  Let $\bar f'$ be the parsimonious extension of $\bar f$ to $T_2(X' \cup q_1
  \cup q_2 \cup q_3)$.  We add $q_1$, $q_2$, and $q_3$ to $Q'$ and set $\bar f =
  \bar f'$.  Assume that the colour of all internal vertices of $P$ is red.
  Then $a_1,a_2,a_3$ are also red.  This in turn implies that $b_1,b_2,b_3$ are
  also red, as otherwise $\beta_{\bar f'}(q_i) > 0$ for some~$i$, and Case~1
  would apply.  We change the colour of the internal vertices of $P$ to blue and
  change the colour of all \markj{adjacent} leaves of $P$ that are leaves of $q_1$, $q_2$
  or $q_3$ to blue.  See \cref{fig:ldiq-case-4.1}.  Since $a_i$ is now coloured
  blue, and $b_i$ red, for each $i \in \{1,2,3\}$, this increases $\beta_{\bar
  f}(Q')$ by at least~$3$.  The only mutation edges introduced into $T_2(X')$
  are the first and last edge of~$P$.  Thus, $\Delta_{\bar f}(T_2(X'))$
  increases by at most~2.  Overall, $\beta_{\bar f}(Q')  - \Delta_{\bar
  f}(T_2(X'))$ increases by at least~1.

  \begin{figure}[t]
    \centering
    \subcaptionbox{Case~4.1\label{fig:ldiq-case-4.1}}{\begin{tikzpicture}
      \path          node [vertex,red node]                               (e)  {}
      ++(0:1)        node [vertex,blue node]                              (f)  {}
      +(270:1)       node [vertex,blue node,label=270:$a_1$]              (a1) {}
      ++(0:1)        node [vertex,blue node]                              (g)  {}
      +(270:1)       node [vertex,blue node,label=270:$a_2$]              (a2) {}
      ++(0:1)        node [vertex,blue node]                              (h)  {}
      +(270:1)       node [vertex,blue node,label=270:$a_3$]              (a3) {}
      ++(0:1)        node [vertex,blue node]                              (i)  {}
      +(270:1)       node [vertex,blue node,label=270:$c_1$]              (c1) {}
      ++(0:1)        node [vertex,red node]                               (j)  {}
      ++(75:1)       node [vertex,red node]                               (k)  {}
      +(135:1)       node [vertex,red node]                               (l)  {}
      ++(15:1)       node [vertex,red node]                               (m)  {}
      +(15:1)        node [vertex,red node]                               (n)  {}
      +(285:1)       node [vertex,red node,label={[xshift=4pt]270:$c_2$}] (c2) {}
      (j) ++(315:1)  node [vertex,red node]                               (o)  {}
      +(225:1)       node [vertex,red node,label=250:$c_3$]               (c3) {}
      ++(315:1)      node [vertex,red node]                               (p)  {}
      +(255:1)       node [vertex,red node]                               (q)  {}
      ++(15:1)       node [vertex,red node]                               (r)  {}
      +(285:1)       node [vertex,red node,label=275:$b_3$]               (b3) {}
      ++(15:1)       node [vertex,red node]                               (s)  {}
      +(90:1)        node [vertex,red node]                               (t)  {}
      ++(345:1)      node [vertex,red node]                               (u)  {}
      +(255:1)       node [vertex,red node,label=270:$d_3$]               (d3) {}
      +(345:1)       node [vertex,red node]                               (v)  {}
      (e) ++(120:1)  node [vertex,red node]                               (aa) {}
      +(210:1)       node [vertex,red node,label=210:$b_1$]               (b1) {}
      ++(120:1)      node [vertex,red node]                               (bb) {}
      +(45:1)        node [vertex,red node]                               (cc) {}
      ++(150:1)      node [vertex,red node]                               (dd) {}
      +(240:1)       node [vertex,red node,label=265:$d_2$]               (d2) {}
      +(150:1)       node [vertex,red node]                               (ee) {}
      (e) ++(240:1)  node [vertex,red node]                               (ff) {}
      ++(180:1)      node [vertex,red node]                               (gg) {}
      +(270:1)       node [vertex,red node,label=270:$b_2$]               (b2) {}
      ++(180:1)      node [vertex,red node]                               (hh) {}
      +(270:1)       node [vertex,red node,label=270:$d_1$]               (d1) {}
      +(180:1)       node [vertex,red node]                               (ii) {}
      (ff) ++(285:1) node [vertex,blue node]                              (jj) {}
      +(315:1)       node [vertex,blue node]                              (kk) {}
      +(255:1)       node [vertex,blue node]                              (ll) {};
      \path [bold edge,red edge] (ee) -- (dd) -- (bb) -- (cc)
      (bb) -- (aa) -- (e) -- (ff) -- (gg) -- (hh) -- (ii)
      (j) -- (k) -- (l)
      (k) -- (m) -- (n) (j) -- (o) -- (p) -- (r) -- (s) -- (t)
      (s) -- (u) -- (v) (p) -- (q);
      \path [bold edge,blue edge] (kk) -- (jj) -- (ll)
      (f) -- (g) -- (h) -- (i);
      \path [bold edge] (ff) -- (jj) (e) -- (f) (i) -- (j);
      \path [thin edge,red edge] (d1) -- (hh) (b2) -- (gg) (b1) -- (aa) (d2) -- (dd)
      (c2) -- (m) (b3) -- (r) (d3) -- (u) (c3) -- (o);
      \path [thin edge,blue edge] (a1) -- (f) (a2) -- (g) (a3) -- (h) (c1) -- (i);
    \end{tikzpicture}}\\[2\bigskipamount]
    \subcaptionbox{Case~4.2\label{fig:ldiq-case-4.2}}{\begin{tikzpicture}
      \path         node [vertex,red node]                                  (e)  {}
      ++(0:1)       node [vertex,red node]                                  (f)  {}
      +(270:1)      node [vertex,red node,label=270:{\vphantom{$d$}$a_1$}]  (a1) {}
      ++(0:1)       node [vertex,red node]                                  (g)  {}
      +(270:1)      node [vertex,red node,label=270:{\vphantom{$d$}$c_1$}]  (c1) {}
      ++(0:1)       node [vertex,blue node]                                 (h)  {}
      +(270:1)      node [vertex,blue node,label=270:$b_1$]                 (b1) {}
      ++(0:1)       node [vertex,blue node]                                 (i)  {}
      +(270:1)      node [vertex,blue node,label=270:{\vphantom{$d$}$a_2$}] (a2) {}
      ++(0:1)       node [vertex,blue node]                                 (j)  {}
      +(270:1)      node [vertex,blue node,label=270:$d_1$]                 (d1) {}
      ++(0:1)       node [vertex,red node]                                  (k)  {}
      ++(75:1)      node [vertex,red node]                                  (l)  {}
      +(345:1)      node [vertex,red node,label={[yshift=2pt]355:$b_2$}]    (b2) {}
      ++(75:1)      node [vertex,red node]                                  (m)  {}
      +(345:1)      node [vertex,red node,label={[yshift=2pt]355:$d_2$}]    (d2) {}
      +(75:1)       node [vertex,red node]                                  (n)  {}
      (k) ++(315:1) node [vertex,blue node]                                 (o)  {}
      +(240:1)      node [vertex,blue node]                                 (p)  {}
      ++(345:1)     node [vertex,blue node]                                 (q)  {}
      +(15:1)       node [vertex,blue node]                                 (r)  {}
      +(315:1)      node [vertex,blue node]                                 (s)  {}
      (e) ++(120:1) node [vertex,red node]                                  (t)  {}
      +(60:1)       node [vertex,red node]                                  (u)  {}
      ++(180:1)     node [vertex,red node]                                  (v)  {}
      +(180:1)      node [vertex,red node]                                  (w)  {}
      +(270:1)      node [vertex,red node,label=270:$c_2$]                  (c2) {}
      (e) ++(255:1) node [vertex,blue node]                                 (aa) {}
      +(225:1)      node [vertex,blue node]                                 (bb) {}
      +(285:1)      node [vertex,blue node]                                 (cc) {};
      \path [bold edge] (g) -- (h) (j) -- (k) -- (o) (e) -- (aa);
      \path [bold edge,blue edge] (p) -- (o) -- (q) (r) -- (q) -- (s)
      (bb) -- (aa) -- (cc) (h) -- (i) -- (j);
      \path [bold edge,red edge] (g) -- (f) -- (e) -- (t) -- (v) -- (w) (t) -- (u)
      (k) -- (l) -- (m) -- (n);
      \path [thin edge,red edge] (c2) -- (v) (a1) -- (f) (c1) --  (g)
      (b2) -- (l) (d2) -- (m);
      \path [thin edge,blue edge] (b1) -- (h) (a2) -- (i) (d1) -- (j);
    \end{tikzpicture}}
    \caption{The updated colouring in Cases~4.1 and~4.2 of the proof of
    \cref{prop:leg-disjoint-quartets}.  Only $T_2(X')$ and the leaves of $q$
    are shown.  Bold solid edges are in $T_2(X')$.  Thin dashed edges are the
    new edges added to $T_2(X' \cup q_1 \cup q_2 \cup q_3)$ (in Case~4.1) or
    to $T_2(X' \cup q_1 \cup q_2)$ (in Case~4.2).}
  \end{figure}

  \paragraph{Case 4.2: A side with an \markj{adjacent} quartet and an \markj{adjacent} leaf.}

  Next suppose that there exists a side $P$ of $T_2(X')$ and two quartets $q_1,
  q_2 \in U$ with $T_1|_{q_1} = a_1b_1|c_1d_1$, $T_1|_{q_2} = a_2b_2|c_2d_2$,
  and such that all leaves of $q_1$ are \markj{adjacent} to $P$ and $a_2$ but
  not $b_2$ is \markj{adjacent} to~$P$.  Assume w.l.o.g.\ that $T_2|_{q_1} =
  a_1c_1|b_1d_1$, and suppose we walk along $P$ such that $a_1$ and $c_1$ appear
  before $b_1$ and~$d_1$.  Assume further that $a_2$ occurs after $a_1$ and
  $c_1$ along~$P$.  (The other case is symmetric using $b_1$ and $d_1$ in place
  of $a_1$ and~$c_1$.) Then let $\bar f'$ be the parsimonious extension of $\bar
  f$ to $T_2(X' \cup q_1 \cup q_2)$.  We add $q_1$ and $q_2$ to $Q'$ and set
  $\bar f = \bar f'$.  Assume that the colour of all internal vertices of $P$ is
  red.  Then we change the colour of all vertices of $q_1$ and $q_2$ that occur
  after $a_1$ and $c_1$ along $P$ to blue, and we colour all internal vertices
  of $P$ that belong to paths between these leaves blue.  See
  \cref{fig:ldiq-case-4.2}.  This increases $\Delta_{\bar f}(T_2(X'))$ by at
  most~2.  At the same time, we obtain $\beta_{\bar f}(q_1) = 2$ and
  $\beta_{\bar f}(q_2) \ge 1$ (as $a_2$ changes colour but $b_2$ does not).
  Thus, $\beta_{\bar f}(Q') $ increases by at least~$3$.  Overall, $\beta_{\bar
  f}(Q') - \Delta_{\bar f}(T_2(X'))$ increases by at least~1.

  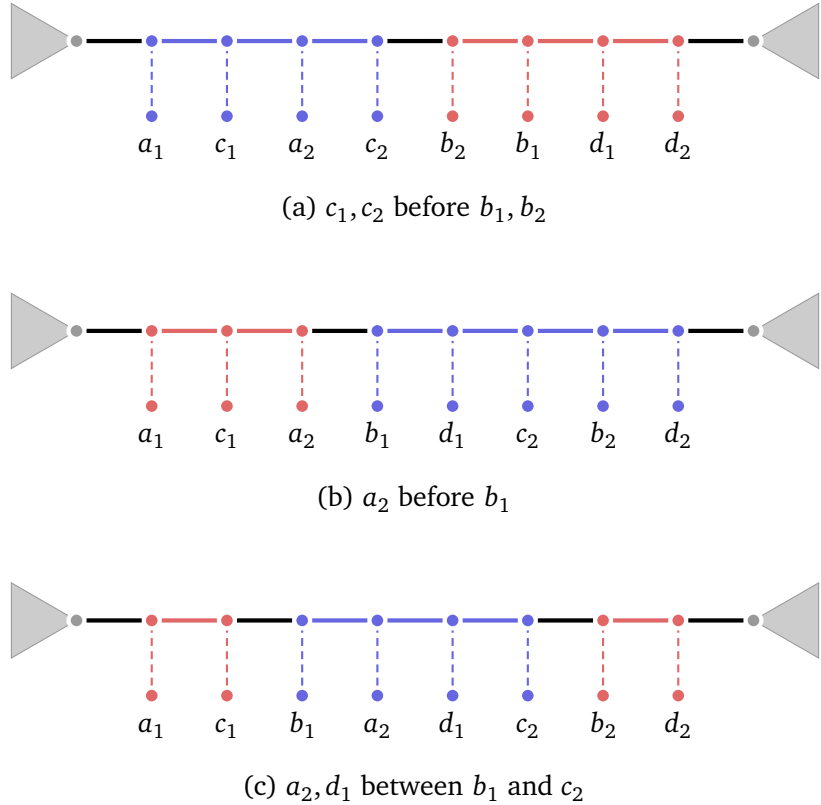
\begin{figure}[b]
    \centering
    \subcaptionbox{$c_1, c_2$ before $b_1, b_2$\label{fig:ldiq-case-4.3.1}}{\begin{tikzpicture}
      \path    node [vertex]                                           (e)  {}
      ++(0:1)  node [vertex,blue node]                                 (f)  {}
      +(270:1) node [vertex,blue node,label=270:{\vphantom{$d$}$a_1$}] (a1) {}
      ++(0:1)  node [vertex,blue node]                                 (g)  {}
      +(270:1) node [vertex,blue node,label=270:{\vphantom{$d$}$c_1$}] (c1) {}
      ++(0:1)  node [vertex,blue node]                                 (h)  {}
      +(270:1) node [vertex,blue node,label=270:{\vphantom{$d$}$a_2$}] (a2) {}
      ++(0:1)  node [vertex,blue node]                                 (i)  {}
      +(270:1) node [vertex,blue node,label=270:{\vphantom{$d$}$c_2$}] (c2) {}
      ++(0:1)  node [vertex,red node]                                  (j)  {}
      +(270:1) node [vertex,red node,label=270:{\vphantom{$d$}$b_2$}]  (b2) {}
      ++(0:1)  node [vertex,red node]                                  (k)  {}
      +(270:1) node [vertex,red node,label=270:{\vphantom{$d$}$b_1$}]  (b1) {}
      ++(0:1)  node [vertex,red node]                                  (l)  {}
      +(270:1) node [vertex,red node,label=270:{\vphantom{$d$}$d_1$}]  (d1) {}
      ++(0:1)  node [vertex,red node]                                  (m)  {}
      +(270:1) node [vertex,red node,label=270:{\vphantom{$d$}$d_2$}]  (d2) {}
      ++(0:1)  node [vertex]                                           (n)  {};
      \begin{scope}[on background layer]
        \path [subtree] (e.center) -- +(150:1) -- +(210:1) -- cycle;
        \path [subtree] (n.center) -- +(30:1) -- +(330:1) -- cycle;
      \end{scope}
      \path [bold edge] (e) -- (f) (i) -- (j) (m) -- (n);
      \path [bold edge,red edge] (j) -- (k) -- (l) -- (m);
      \path [bold edge,blue edge] (f) -- (g) -- (h) -- (i);
      \path [thin edge,red edge] (b2) -- (j) (b1) -- (k) (d1) -- (l) (d2) -- (m);
      \path [thin edge,blue edge] (a1) -- (f) (c1) -- (g) (a2) -- (h) (c2) -- (i);
    \end{tikzpicture}}\\[2\bigskipamount]
    \subcaptionbox{$a_2$ before $b_1$\label{fig:ldiq-case-4.3.2}}{\begin{tikzpicture}
      \path    node [vertex]                                           (e)  {}
      ++(0:1)  node [vertex,red node]                                  (f)  {}
      +(270:1) node [vertex,red node,label=270:{\vphantom{$d$}$a_1$}]  (a1) {}
      ++(0:1)  node [vertex,red node]                                  (g)  {}
      +(270:1) node [vertex,red node,label=270:{\vphantom{$d$}$c_1$}]  (c1) {}
      ++(0:1)  node [vertex,red node]                                  (h)  {}
      +(270:1) node [vertex,red node,label=270:{\vphantom{$d$}$a_2$}]  (a2) {}
      ++(0:1)  node [vertex,blue node]                                 (i)  {}
      +(270:1) node [vertex,blue node,label=270:{\vphantom{$d$}$b_1$}] (b1) {}
      ++(0:1)  node [vertex,blue node]                                 (j)  {}
      +(270:1) node [vertex,blue node,label=270:{\vphantom{$d$}$d_1$}] (d1) {}
      ++(0:1)  node [vertex,blue node]                                 (k)  {}
      +(270:1) node [vertex,blue node,label=270:{\vphantom{$d$}$c_2$}] (c2) {}
      ++(0:1)  node [vertex,blue node]                                 (l)  {}
      +(270:1) node [vertex,blue node,label=270:{\vphantom{$d$}$b_2$}] (b2) {}
      ++(0:1)  node [vertex,blue node]                                 (m)  {}
      +(270:1) node [vertex,blue node,label=270:{\vphantom{$d$}$d_2$}] (d2) {}
      ++(0:1)  node [vertex]                                           (n)  {};
      \begin{scope}[on background layer]
        \path [subtree] (e.center) -- +(150:1) -- +(210:1) -- cycle;
        \path [subtree] (n.center) -- +(30:1) -- +(330:1) -- cycle;
      \end{scope}
      \path [bold edge] (e) -- (f) (h) -- (i) (m) -- (n);
      \path [bold edge,red edge] (f) -- (g) -- (h);
      \path [bold edge,blue edge] (i) -- (j) -- (k) -- (l) -- (m);
      \path [thin edge,red edge] (a1) -- (f) (c1) -- (g) (a2) -- (h);
      \path [thin edge,blue edge] (b1) -- (i) (d1) -- (j) (c2) -- (k)
      (b2) -- (l) (d2) -- (m);
    \end{tikzpicture}}\\[2\bigskipamount]
    \subcaptionbox{$a_2, d_1$ between $b_1$ and $c_2$\label{fig:ldiq-case-4.3.3}}{\begin{tikzpicture}
      \path    node [vertex]                                           (e)  {}
      ++(0:1)  node [vertex,red node]                                  (f)  {}
      +(270:1) node [vertex,red node,label=270:{\vphantom{$d$}$a_1$}]  (a1) {}
      ++(0:1)  node [vertex,red node]                                  (g)  {}
      +(270:1) node [vertex,red node,label=270:{\vphantom{$d$}$c_1$}]  (c1) {}
      ++(0:1)  node [vertex,blue node]                                 (h)  {}
      +(270:1) node [vertex,blue node,label=270:{\vphantom{$d$}$b_1$}] (b1) {}
      ++(0:1)  node [vertex,blue node]                                 (i)  {}
      +(270:1) node [vertex,blue node,label=270:{\vphantom{$d$}$a_2$}] (a2) {}
      ++(0:1)  node [vertex,blue node]                                 (j)  {}
      +(270:1) node [vertex,blue node,label=270:{\vphantom{$d$}$d_1$}] (d1) {}
      ++(0:1)  node [vertex,blue node]                                 (k)  {}
      +(270:1) node [vertex,blue node,label=270:{\vphantom{$d$}$c_2$}] (c2) {}
      ++(0:1)  node [vertex,red node]                                  (l)  {}
      +(270:1) node [vertex,red node,label=270:{\vphantom{$d$}$b_2$}]  (b2) {}
      ++(0:1)  node [vertex,red node]                                  (m)  {}
      +(270:1) node [vertex,red node,label=270:{\vphantom{$d$}$d_2$}]  (d2) {}
      ++(0:1)  node [vertex]                                           (n)  {};
      \begin{scope}[on background layer]
        \path [subtree] (e.center) -- +(150:1) -- +(210:1) -- cycle;
        \path [subtree] (n.center) -- +(30:1) -- +(330:1) -- cycle;
      \end{scope}
      \path [bold edge] (e) -- (f) (g) -- (h) (k) -- (l) (m) -- (n);
      \path [bold edge,red edge] (f) -- (g) (l) -- (m);
      \path [bold edge,blue edge] (h) -- (i) -- (j) -- (k);
      \path [thin edge,red edge] (a1) -- (f) (c1) -- (g) (b2) -- (l) (d2) -- (m);
      \path [thin edge,blue edge] (b1) -- (h) (a2) -- (i) (d1) -- (j) (c2) -- (k);
    \end{tikzpicture}}
    \caption{The updated colouring in Case~4.3 of the proof of
    \cref{prop:leg-disjoint-quartets}.  Only $T_2(X')$ and the leaves of $q$
    are shown.  Bold solid edges are in $T_2(X')$.  Thin dashed edges are the
    new edges added to $T_2(X' \cup q_1 \cup q_2)$.}
  \end{figure}

  \paragraph{Case 4.3: A side with two \markj{adjacent} quartets.}

  The final case we consider is when there are two quartets $q_1, q_2 \in U$
  such that all leaves of $q_1$ and $q_2$ are \markj{adjacent} to the same side
  $P$ of $T_2(X')$.  Assume w.l.o.g., that $T_1|_{q_1} = a_1b_1|c_1d_1$,
  $T_1|_{q_2} = a_2b_2|c_2d_2$, $T_2|_{q_1} = a_1c_1|b_1d_1$, and $T_2|_{q_2} =
  a_2c_2|b_2d_2$.  Assume further that the leaves of $q_1$ occur in the order
  $a_1, c_1, b_1, d_1$ along $P$ and, following $P$ in the same direction, the
  leaves of $q_2$ occur in the order $a_2, c_2, b_2, d_2$ along $P$.  (If $c_1$
  occurs before $a_1$, then we may swap the roles of $a_1$ and $c_1$ in the
  argument that follows; similarly for the pairs $(b_1,d_1), (a_2,c_2),
  (b_2,d_2)$.)

  If both $c_1$ and $c_2$ occur before both $b_1$ and $b_2$ along~$P$, then let
  $\bar f'$ be the parsimonious extension of $\bar f$ to $T_2(X' \cup q_1 \cup
  q_2)$.  We add $q_1$ and $q_2$ to $Q'$ and set $\bar f = \bar f'$.  Assume
  that the colour of all internal vertices of $P$ is red.  We change the colour
  of $a_1, c_1, a_2, c_2$ to blue and also change the colour of all vertices on
  the paths between these four leaves to blue.  See \cref{fig:ldiq-case-4.3.1}.
  This increases $\Delta_{\bar f}(T_2(X'))$ by at most 2 and ensures that
  $\beta_{\bar f}(q_1) = \beta_{\bar f}(q_2) = 2$.  Thus, $\beta_{\bar f}(Q') -
  \Delta_{\bar f}(T_2(X'))$ increases by at least~2.

  If $c_1$ and $c_2$ do not both occur before $b_1$ and $b_2$, then the four
  leaves $b_1, c_1, b_2, c_2$ must occur in the order $c_1,b_1,c_2,b_2$ or
  $c_2,b_2,c_1,b_1$ along~$P$. Assume that the order is $c_1,b_1,c_2,b_2$ (the
  other case is symmetric).  Then observe that $a_1$ occurs before $c_1$ and
  $d_2$ occurs after~$b_2$.  \markj{Thus, these six leaves occur in the order
  $a_1,c_1, b_1, c_2, b_2, d_2$.} We distinguish the possible positions of the
  two leaves $a_2$ and $d_1$ and in each case update $\bar f$ so that
  $\beta_{\bar f}(Q')  - \Delta_{\bar f}(T_2(X'))$ increases by at least~1.  In
  each case, the starting point is the parsimonious extension $\bar f'$ of $\bar
  f$ to $T_2(X' \cup q_1 \cup q_2)$.  We assume that $\bar f'$ colours all
  vertices on $P$ and all leaves of $q_1$ and $q_2$ red.

  If $a_2$ occurs before $b_1$, then we change the colours of $b_1$, $d_1$,
  $c_2$, $b_2$, and $d_2$ and the colours of all vertices on the paths between
  them in $T_2$ to blue.  See \cref{fig:ldiq-case-4.3.2}.  This ensures that
  $\beta_{\bar f}(q_1) = 2$ and $\beta_{\bar f}(q_2) = 1$.  Thus, $\beta_{\bar
  f}(Q')$ increases by~$3$.  At the same time, we introduce at most two mutation
  edges into~$P$, so $\beta_{\bar f}(Q') - \Delta_{\bar f}(T_2(X'))$ increases
  by at least~1.

  The case when $d_1$ occurs after $c_2$ is analogous to the case when $a_2$
  occurs before~$b_1$. \markj{We colour $a_1,c_1,b_1,a_2,c_2$, and all vertices
  on the paths between them blue.} \nz{This ensures that $\beta_{\bar f}(q_1) =
  1$ and $\beta_{\bar f}(q_2) = 2$ and introduces at most two mutation
  edges into~$P$. Thus, $\beta_{\bar f}(Q') - \Delta_{\bar f}(T_2(X'))$
  increases by at least~1.}

  This leaves the case when both $a_2$ and $d_1$ occur between $b_1$ and~$c_2$.
  In this case, we change the colour of $b_1, d_1, a_2, c_2$, and of all
  vertices on the paths between them in $T_2$ to blue.  See
  \cref{fig:ldiq-case-4.3.3}.  This ensures that $\beta_{\bar f}(q_1) =
  \beta_{\bar f}(q_2) = 2$ and $\Delta_{\bar f}(T_2(X'))$ increases by~2.  Thus,
  $\beta_{\bar f}(Q')  - \Delta_{\bar f}(T_2(X'))$ increases by~2.

  \bigskip

  Once none of these cases is applicable, we obtain a subset $Q' \subseteq Q$
  and a colouring $\bar f$ of $T_2(X')$ such that $\beta_{\bar f}(Q')  -
  \Delta_{\bar f}(T_2(X')) \ge \frac{|Q'|}{3}$.  It remains to prove that $|Q'|
  \ge \frac{|Q|}{9}$.  Since $Q = Q' \cup U$, this follows if we can prove that
  $|U| \le 8|Q'|$.

  Consider any quartet $q \in U$ with $T_1|_q = ab|cd$.  As argued before, once
  none of Cases 1--3 applies, either all leaves of $q$ are \markj{adjacent} to
  one side of $T_2(X')$, or w.l.o.g., $a$ and $b$ are \markj{adjacent} to
  different sides of $T_2(X')$, because $q$ is incompatible.  We partition $U$
  into two subsets $U_1$ and $U_2$, containing the quartets in $U$ whose leaves
  are all \markj{adjacent} to the same side and those whose leaves are
  \markj{adjacent} to at least two sides, respectively.

  Now we charge the quartets in $U$ to the sides of $T_2(X')$.  We charge each
  quartet $q \in U_1$ to the side of $T_2(X')$ to which the leaves of $q$ are
  \markj{adjacent}.  We charge each quartet $q \in U_2$ to the \emph{two} sides
  of $T_2(X')$ to which $a$ and $b$ are \markj{adjacent}.

  Since Case~4.3 does not apply to the quartets in $U$, there is no side that is
  charged for more than one quartet in $U_1$.  Since Case~4.1 does not apply,
  there is no side that is charged for more than two quartets in $U_2$.  Since
  Case~4.2 does not apply, there is no side that is charged for a quartet in
  $U_1$ and for at least one quartet in $U_2$.  Thus, every side of $T_2(X')$
  that is charged for any quartet is charged for one quartet in $U_1$ or for at
  most two quartets in $U_2$.  Since every quartet in $U_1$ is charged to one
  side of $T_2(X')$, and every quartet in $U_2$ is charged to two sides of
  $T_2(X')$, the number of sides of $T_2(X')$ is thus at least $|U_1| + |U_2| =
  |U|$.  On the other hand, since $T_2(X')$ has $|X'| = 4|Q'|$ leaves, it has at
  most $2|X'| = 8|Q'|$ sides.  Thus, $|U| \le 8|Q'|$.  This finishes the proof
  that $|Q'| \ge \frac{|Q|}{9}$ and thus the proof of the proposition.
\end{proof}

\subsection{Finding Leg-Disjoint Incompatible Quartets}

\label{sec:finding-quartets}


It remains to find a set of leg-disjoint incompatible quartets of size at least
$\frac{\dtbr(T_1, T_2)}{2(\lg n + 1)}$, where $n = |X|$.  \markj{In combination
with \cref{prop:leg-disjoint-quartets}, this implies that $\dmp^t(T_1,T_2)\geq
\frac{\dtbr(T_1, T_2)}{54(\lg n + 1)}$, as claimed in \cref{thm:lower-bound}.}
To do this, we use an ILP formulation of the unrooted MAF problem by Van Wersch
et al.\ \cite{werschReflectionsKernelizingComputing2020}.  For a pair of trees
$(T_1, T_2)$ on $X$, let $Q$ be the set of incompatible quartets of $T_1$ and
$T_2$.  For a quartet $q \in Q$, let $\L(q)$ be the set of edges of $T_1$ that
belong to the legs of~$q$.  Van Wersch et al.\ proved that the following ILP
expresses the unrooted MAF problem, where $E_1$ is the set of edges of $T_1$ and
$x_e \in \{0, 1\}$ indicates whether we include $e$ in a set of edges we cut to
obtain an AF of $(T_1, T_2)$:

\begin{equation}
  \begin{gathered}
    \text{Minimize}\ \sum_{e \in E_1} x_e\\
    \begin{aligned}
      \text{s.t.}\ \sum_{e \in \L(q)} x_e &\ge 1 && \forall q \in Q\\
      x_e &\in \{0, 1\} && \forall e \in E_1.
    \end{aligned}
  \end{gathered}
  \label{eq:primal}
\end{equation}

The constraints express that we obtain an AF of $(T_1, T_2)$ by cutting a subset
of edges in $T_1$ that contains at least one edge in $\L(q)$ for every
incompatible quartet $q \in Q$.  For the remainder of this section, we say that
an edge set $E'$ \emph{hits} a quartet $q \in Q$ if $E'$ contains at least one
edge in $\L(q)$.  The objective function expresses the goal to cut as few edges
as possible, to obtain an MAF\@.  Recall that the number of edges cut to produce
a MAF of $T_1$ and $T_2$ is exactly $\dtbr(T_1, T_2)$, so the objective function
value of any feasible solution of \cref{eq:primal} is an upper bound on
$\dtbr(T_1, T_2)$.

Interestingly, the integral version of the dual of this LP corresponds
to choosing a subset of quartets from $Q$ that are pairwise leg-disjoint:

\begin{equation}
  \begin{gathered}
    \text{Maximize}\ \sum_{q \in Q} y_q\\
    \begin{aligned}
      \text{s.t.}\ \sum_{q \in Q: e \in \L(q)} y_q &\le 1 && \forall e \in E_1\\
      y_q &\in \{0, 1\} && \forall q \in Q.
    \end{aligned}
  \end{gathered}
  \label{eq:dual}
\end{equation}

Indeed, the variable $y_q$ for each quartet $q \in Q$ indicates whether it is
chosen. The constraints in \cref{eq:dual} ensure that no edge of $T_1$ is
included in the legs of more than one chosen quartet, that is, all quartets are
leg-disjoint.  This observation motivates our approach in the proof of the
following proposition.

\begin{prop}
  \label{prop:finding-leg-disjoint-quartets}
  For any pair of trees $T_1$ and $T_2$ on $X$, there exists a set $Q'$ of pairwise
  leg-disjoint incompatible quartets such that $|Q'|\geq
  \frac{\dtbr(T_1,T_2)}{2(\lg n + 1)}$, where $n = |X|$.
\end{prop}

\begin{proof}
  Our goal is to find feasible solutions $\hat x$ and $\hat y$ of the ILPs
  \labelcref{eq:primal,eq:dual} such that $\sum_{e \in E_1} \hat x_e \le 2(\lg n
  + 1) \cdot \sum_{q \in Q} \hat y_q$.  In other words, we want to find a subset
  $E' = \{ e \in E_1 \mid \hat x_e = 1 \}$ that hits all quartets in $Q$ and a
  subset $Q' = \{ q \in Q \mid \hat y_q = 1 \}$ of leg-disjoint quartets in $Q$
  such that $|E'| \le |Q'| \cdot 2(\lg n + 1)$.  As $|E'|$ gives an upper
  bound on $\dtbr(T_1,T_2)$, this implies that $|Q'|\geq
  \frac{\dtbr(T_1,T_2)}{2(\lg n + 1)}$, as required.

  To describe the choice of quartets in $Q'$ and edges in $E'$, we need a bit of
  notation.  Let $E'$ be some set of edges in $T_1$ that we have selected at
  some point in the algorithm, and consider a quartet $q \in Q$ that is not hit
  by $E'$.  Throughout this section, we refer to such a quartet $q$ with $T_1|_q
  = ab|cd$ as the quartet $ab|cd$, we use $P_{ab}$ to denote the leg with
  endpoints $a$ and $b$, and we use $P_{cd}$ to denote the leg with endpoints
  $c$ and~$d$. $u_{ab}$ and $u_{cd}$ are the joints of $q$ included in $P_{ab}$
  and~$P_{cd}$, respectively.  Let $X_a$ be the set of all leaves reachable from
  $a$ via paths in $T_1$ that do not include $u_{ab}$ nor any edges in~$E'$.  We
  define sets $X_b$, $X_c$, and $X_d$ analogously.  Let $e_a$ be the first edge
  on the path from $u_{ab}$ to $a$ in $T_1$, and let $e_b$ be the first edge on
  the path from $u_{ab}$ to $b$ in $T_1$.  These definitions are illustrated in
  \cref{fig:exampleEmbeddedQuartet}.  Note that the sets $X_a, X_b, X_c, X_d$
  depend on the choice of $E'$ as well as $q$; when we need to specify $E'$ or
  $q$ (for instance, when a leaf is part of two different quartets under
  consideration or we refer to the states of $E'$ before and after an update),
  we will denote these sets by $X_a^{E'}, X_b^{E'}, X_c^{E'}, X_d^{E'}$ or
  $X_a^q, X_b^q, X_c^q, X_d^q$, or  $X_a^{E',q}, X_b^{E',q}, X_c^{E',q},
  X_d^{E',q}$ when we need to specify both. We use $F_1$ to denote the forest
  obtained from $T_1$ by cutting the edges in $E'$, suppressing degree-$2$
  vertices, and deleting unlabelled vertices of degree less than~$2$.  When it
  is necessary to specify the set of edges $E'$ cut to obtain $F_1$, we refer to
  $F_1$ as $F_1^{E'}$.  Every edge $e \in F_1$ corresponds to a path between its
  endpoints in $T_1$.  We refer to this path as $P_e$.  For two vertices $a$ and
  $b$ in the same connected component of $F_1$, we use $\tilde P_{ab}$ to refer
  to the path from $a$ to $b$ in~$F_1$.  Note that this implies that $P_{ab} =
  \bigcup_{e \in \tilde P_{ab}} P_e$.

  \begin{figure}[t]
    \centering
    \begin{tikzpicture}
      \path             node [vertex,red node,label=185:$u_{ab}$]                     (uab) {}
      ++(15:1)          node [vertex,red node]                                        (e)   {}
      ++(345:1)         node [vertex,red node]                                        (f)   {}
      +(270:1)          node [vertex]                                                 (k)   {}
      ++(15:1)          node [vertex,red node]                                        (g)   {}
      +(90:1)           node [vertex]                                                 (l)   {}
      ++(345:1)         node [vertex,red node,label={[yshift=3pt]355:$u_{cd}$}]       (ucd) {}
      (e) ++(90:1)      node [vertex]                                                 (h)   {}
      +(120:1)          node [vertex]                                                 (i)   {}
      +(60:1)           node [vertex]                                                 (j)   {}
      (uab) ++(135:1)   node [vertex,red node]                                        (m)   {}
      +(60:1)           node [vertex]                                                 (p)   {}
      ++(165:1)         node [vertex,in subtree,red node]                             (n)   {}
      +(135:1)          node [vertex,in subtree,red node,label=135:$a$]               (a)   {}
      +(240:1)          node [vertex,in subtree]                                      (o)   {}
      (uab) ++(255:1.2) node [vertex,red node]                                        (q)   {}
      +(330:1)          node [vertex]                                                 (s)   {}
      ++(225:1)         node [vertex,in subtree,red node]                             (r)   {}
      +(255:1)          node [vertex,in subtree,red node,label={[xshift=2pt]265:$b$}] (b)   {}
      ++(165:1)         node [vertex,in subtree]                                      (t)   {}
      +(135:1)          node [vertex]                                                 (u)   {}
      +(195:1)          node [vertex,in subtree]                                      (v)   {}
      (ucd) ++(60:1.25) node [vertex,in subtree,red node]                             (aa)  {}
      ++(30:1)          node [vertex,in subtree,red node]                             (bb)  {}
      +(60:1)           node [vertex,in subtree,red node,label=70:$c$]                (c)   {}
      +(315:1)          node [vertex,in subtree]                                      (cc)  {}
      (aa) ++(120:1)    node [vertex,in subtree]                                      (dd)  {}
      +(150:1)          node [vertex,in subtree]                                      (ee)  {}
      +(90:1)           node [vertex,in subtree]                                      (ff)  {}
      (ucd) ++(300:1)   node [vertex,red node]                                        (gg)  {}
      +(225:1)          node [vertex]                                                 (jj)  {}
      ++(330:1)         node [vertex,in subtree,red node]                             (hh)  {}
      +(30:1)           node [vertex,in subtree]                                      (ll)  {}
      ++(300:1)         node [vertex,in subtree,red node]                             (ii)  {}
      +(240:1)          node [vertex,in subtree]                                      (kk)  {}
      +(330:1)          node [vertex,in subtree,red node,label={[yshift=2pt]355:$d$}] (d)   {};
      \path [bold edge,red edge]
      (uab) -- (e) (f) -- (g)
      (uab) to node [above right,yshift=-2pt,xshift=-2pt] {$e_a$} (m) -- (n) -- (a)
      (uab) to node [right,yshift=-2pt] {$e_b$} (q) -- (r) -- (b)
      (ucd) to node [above left,xshift=2pt,yshift=-2pt,pos=0.25] {$e_c$} (aa) -- (bb) -- (c)
      (ucd) to node [below left,xshift=2pt,yshift=2pt] {$e_d$} (gg) -- (hh) -- (ii) -- (d);
      \path [bold edge,red edge,densely dashed] (e) -- (f) (g) -- (ucd);
      \path [thin edge] (m) -- (p) (q) -- (s) (t) -- (u) (gg) -- (jj);
      \path [thin edge,solid] (n) -- (o) (r) -- (t) -- (v) (e) -- (h) (i) -- (h) -- (j)
      (f) -- (k) (g) -- (l) (bb) -- (cc) (aa) -- (dd) (ee) -- (dd) -- (ff)
      (ii) -- (kk) (hh) -- (ll);
      \begin{scope}[overlay]
        \path [name path=helpers]
        let \p1 = ($(a.center) - (n.center)$),
            \p2 = ($(o.center) - (a.center)$),
            \p3 = ($(n.center) - (o.center)$)
        in  (n) -- +(\y1,-\x1)
            (a) -- +(\y1,-\x1)
            (a) -- +(\y2,-\x2)
            (o) -- +(\y2,-\x2)
            (o) -- +(\y3,-\x3)
            (n) -- +(\y3,-\x3);
        \path [name path=a circle] (a.center) circle [radius=6mm];
        \path [name path=n circle] (n.center) circle [radius=6mm];
        \path [name path=o circle] (o.center) circle [radius=6mm];
        \path [name intersections={of=a circle and helpers}]
        coordinate (ar) at (intersection-1)
        coordinate (al) at (intersection-2);
        \path [name intersections={of=n circle and helpers}]
        coordinate (nr) at (intersection-1)
        coordinate (nb) at (intersection-2);
        \path [name intersections={of=o circle and helpers}]
        coordinate (ol) at (intersection-1)
        coordinate (ob) at (intersection-2);
      \end{scope}
      \begin{scope}[on background layer]
        \path [subtree]
        let \p1 = ($(ar) - (a.center)$),
            \p2 = ($(al) - (a.center)$),
            \p3 = ($(ol) - (o.center)$),
            \p4 = ($(ob) - (o.center)$),
            \p5 = ($(nb) - (n.center)$),
            \p6 = ($(nr) - (n.center)$),
            \n1 = {atan2(\y1,\x1)},
            \n2 = {atan2(\y2,\x2)+360},
            \n3 = {atan2(\y3,\x3)},
            \n4 = {atan2(\y4,\x4)},
            \n5 = {atan2(\y5,\x5)},
            \n6 = {atan2(\y6,\x6)}
        in  (ol) arc [start angle=\n3,end angle=\n4,radius=6mm] -- 
            (nb) arc [start angle=\n5,end angle=\n6,radius=6mm] --
            (ar) arc [start angle=\n1,end angle=\n2,radius=6mm] -- cycle;
      \end{scope}
      \begin{scope}[overlay]
        \path [name path=helpers]
        let \p1 = ($(v.center) - (t.center)$),
            \p2 = ($(t.center) - (r.center)$),
            \p3 = ($(b.center) - (v.center)$),
            \p4 = ($(r.center) - (b.center)$)
        in  (r) -- +(\y4,-\x4)
            (b) -- +(\y4,-\x4)
            (r) -- +(\y2,-\x2)
            (t) -- +(\y2,-\x2)
            (t) -- +(\y1,-\x1)
            (v) -- +(\y1,-\x1)
            (v) -- +(\y3,-\x3)
            (b) -- +(\y3,-\x3);
        \path [name path=b circle] (b.center) circle [radius=6mm];
        \path [name path=r circle] (r.center) circle [radius=6mm];
        \path [name path=t circle] (t.center) circle [radius=6mm];
        \path [name path=v circle] (v.center) circle [radius=6mm];
        \path [name intersections={of=b circle and helpers}]
        coordinate (bl) at (intersection-1)
        coordinate (br) at (intersection-2);
        \path [name intersections={of=r circle and helpers}]
        coordinate (rt) at (intersection-1)
        coordinate (rr) at (intersection-2);
        \path [name intersections={of=t circle and helpers}]
        coordinate (tr) at (intersection-1)
        coordinate (tl) at (intersection-2);
        \path [name intersections={of=v circle and helpers}]
        coordinate (vt) at (intersection-1)
        coordinate (vl) at (intersection-2);
      \end{scope}
      \begin{scope}[on background layer]
        \path [subtree]
        let \p1 = ($(br) - (b.center)$),
            \p2 = ($(bl) - (b.center)$),
            \p3 = ($(rr) - (r.center)$),
            \p4 = ($(rt) - (r.center)$),
            \p5 = ($(tr) - (t.center)$),
            \p6 = ($(tl) - (t.center)$),
            \p7 = ($(vt) - (v.center)$),
            \p8 = ($(vl) - (v.center)$),
            \n1 = {atan2(\y1,\x1)},
            \n2 = {atan2(\y2,\x2)},
            \n3 = {atan2(\y3,\x3)},
            \n4 = {atan2(\y4,\x4)},
            \n5 = {atan2(\y5,\x5)},
            \n6 = {atan2(\y6,\x6)},
            \n7 = {atan2(\y7,\x7)},
            \n8 = {atan2(\y8,\x8)+360}
        in  (bl) arc [start angle=\n2,end angle=\n1,radius=6mm] --
            (rr) arc [start angle=\n3,end angle=\n4,radius=6mm] --
            (tr) arc [start angle=\n5,end angle=\n6,radius=6mm] --
            (vt) arc [start angle=\n7,end angle=\n8,radius=6mm] -- cycle;
      \end{scope}
      \begin{scope}[overlay]
        \path [name path=helpers]
        let \p1 = ($(cc.center) - (aa.center)$),
            \p2 = ($(c.center) - (cc.center)$),
            \p3 = ($(ff.center) - (c.center)$),
            \p4 = ($(ee.center) - (ff.center)$),
            \p5 = ($(aa.center) - (ee.center)$)
        in  (aa) -- +(\y1,-\x1)
            (cc) -- +(\y1,-\x1)
            (cc) -- +(\y2,-\x2)
            (c)  -- +(\y2,-\x2)
            (c)  -- +(\y3,-\x3)
            (ff) -- +(\y3,-\x3)
            (ff) -- +(\y4,-\x4)
            (ee) -- +(\y4,-\x4)
            (ee) -- +(\y5,-\x5)
            (aa) -- +(\y5,-\x5);
        \path [name path=aa circle] (aa.center) circle [radius=6mm];
        \path [name path=cc circle] (cc.center) circle [radius=6mm];
        \path [name path=c circle]  (c.center)  circle [radius=6mm];
        \path [name path=ee circle] (ee.center) circle [radius=6mm];
        \path [name path=ff circle] (ff.center) circle [radius=6mm];
        \path [name intersections={of=aa circle and helpers}]
        coordinate (aab) at (intersection-1)
        coordinate (aal) at (intersection-2);
        \path [name intersections={of=cc circle and helpers}]
        coordinate (ccr) at (intersection-1)
        coordinate (ccb) at (intersection-2);
        \path [name intersections={of=c circle and helpers}]
        coordinate (cr) at (intersection-1)
        coordinate (ct) at (intersection-2);
        \path [name intersections={of=ff circle and helpers}]
        coordinate (ffr) at (intersection-1)
        coordinate (ffl) at (intersection-2);
        \path [name intersections={of=ee circle and helpers}]
        coordinate (eet) at (intersection-1)
        coordinate (eel) at (intersection-2);
      \end{scope}
      \begin{scope}[on background layer]
        \path [subtree]
        let \p1    = ($(aal) - (aa.center)$),
            \p2    = ($(aab) - (aa.center)$),
            \p3    = ($(ccb) - (cc.center)$),
            \p4    = ($(ccr) - (cc.center)$),
            \p5    = ($(cr)  - (c.center)$),
            \p6    = ($(ct)  - (c.center)$),
            \p7    = ($(ffr) - (ff.center)$),
            \p8    = ($(ffl) - (ff.center)$),
            \p9    = ($(eet) - (ee.center)$),
            \p{10} = ($(eel) - (ee.center)$),
            \n1    = {atan2(\y1,\x1)},
            \n2    = {atan2(\y2,\x2)},
            \n3    = {atan2(\y3,\x3)},
            \n4    = {atan2(\y4,\x4)},
            \n5    = {atan2(\y5,\x5)},
            \n6    = {atan2(\y6,\x6)},
            \n7    = {atan2(\y7,\x7)},
            \n8    = {atan2(\y8,\x8)},
            \n9    = {atan2(\y9,\x9)},
            \n{10} = {atan2(\y{10},\x{10})+360}
        in (aal) arc [start angle=\n1,end angle=\n2,radius=6mm] --
           (ccb) arc [start angle=\n3,end angle=\n4,radius=6mm] --
           (cr)  arc [start angle=\n5,end angle=\n6,radius=6mm] --
           (ffr) arc [start angle=\n7,end angle=\n8,radius=6mm] --
           (eet) arc [start angle=\n9,end angle=\n{10},radius=6mm] -- cycle;
      \end{scope}
      \begin{scope}[overlay]
        \path [name path=helpers]
        let \p1 = ($(kk.center) - (hh.center)$),
            \p2 = ($(d.center) - (kk.center)$),
            \p3 = ($(ll.center) - (d.center)$),
            \p4 = ($(hh.center) - (ll.center)$)
        in  (hh) -- +(\y1,-\x1)
            (kk) -- +(\y1,-\x1)
            (kk) -- +(\y2,-\x2)
            (d)  -- +(\y2,-\x2)
            (d)  -- +(\y3,-\x3)
            (ll) -- +(\y3,-\x3)
            (ll) -- +(\y4,-\x4)
            (hh) -- +(\y4,-\x4);
        \path [name path=hh circle] (hh.center) circle [radius=6mm];
        \path [name path=kk circle] (kk.center) circle [radius=6mm];
        \path [name path=d circle]  (d.center)  circle [radius=6mm];
        \path [name path=ll circle] (ll.center) circle [radius=6mm];
        \path [name intersections={of=hh circle and helpers}]
        coordinate (hhl) at (intersection-1)
        coordinate (hht) at (intersection-2);
        \path [name intersections={of=kk circle and helpers}]
        coordinate (kkl) at (intersection-1)
        coordinate (kkb) at (intersection-2);
        \path [name intersections={of=d circle and helpers}]
        coordinate (dr) at (intersection-1)
        coordinate (db) at (intersection-2);
        \path [name intersections={of=ll circle and helpers}]
        coordinate (llr) at (intersection-1)
        coordinate (llt) at (intersection-2);
      \end{scope}
      \begin{scope}[on background layer]
        \path [subtree]
        let \p1 = ($(kkl) - (kk.center)$),
            \p2 = ($(kkb) - (kk.center)$),
            \p3 = ($(db)  - (d.center)$),
            \p4 = ($(dr)  - (d.center)$),
            \p5 = ($(llr) - (ll.center)$),
            \p6 = ($(llt) - (ll.center)$),
            \p7 = ($(hht) - (hh.center)$),
            \p8 = ($(hhl) - (hh.center)$),
            \n1 = {atan2(\y1,\x1)},
            \n2 = {atan2(\y2,\x2)+360},
            \n3 = {atan2(\y3,\x3)},
            \n4 = {atan2(\y4,\x4)},
            \n5 = {atan2(\y5,\x5)},
            \n6 = {atan2(\y6,\x6)},
            \n7 = {atan2(\y7,\x7)},
            \n8 = {atan2(\y8,\x8)}
        in (kkl) arc [start angle=\n1,end angle=\n2,radius=6mm] --
           (db)  arc [start angle=\n3,end angle=\n4,radius=6mm] --
           (llr) arc [start angle=\n5,end angle=\n6,radius=6mm] --
           (hht) arc [start angle=\n7,end angle=\n8,radius=6mm] -- cycle;
      \end{scope}
    \end{tikzpicture}
    \caption{Example showing the vertices $u_{ab}, u_{cd}$, edges
    $e_a,e_b,e_c,e_d$, and sets $X_a,X_b,X_c,X_d$, for some quartet $q = ab|cd$,
    and a subset of edges $E'$ of $T_1$ that does not hit the legs of~$q$.
    Dashed edges are edges in~$E'$.  Bold red edges are those in~$T_1(q)$.}
    \label{fig:exampleEmbeddedQuartet}
  \end{figure}
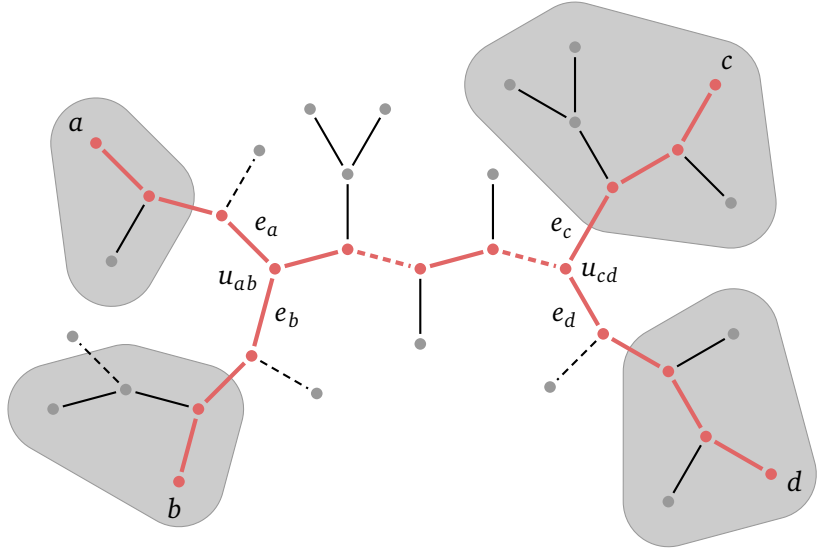

  To construct $E'$ and $Q'$, we use a simple greedy algorithm: We start by
  setting $E' = \emptyset$ and $Q' = \emptyset$.  We maintain the invariant that
  $|E'| \le |Q'| \cdot 2(\lg n + 1)$ and that $Q' \subseteq Q$ is a subset of
  leg-disjoint quartets.  Thus, once $E'$ hits all quartets in $Q$, we obtain
  the desired sets $Q'$ and $E'$.

  As long as there exists a quartet $q \in Q$ that is not being hit by $E'$ yet,
  we choose such a quartet $q$, add a subset of the edges in $\L(q)$ to $E'$,
  and add $q$ to $Q'$.  We choose the quartet $q = ab|cd$ that minimizes
  $\Bigl|X_a^{E',q} \cup X_b^{E',q}\Bigr|$, where we assume that
  $\Bigl|X_a^{E',q} \cup X_b^{E',q}\Bigr| \le \Bigl|X_c^{E',q} \cup
  X_d^{E',q}\Bigr|$. Among all such quartets, we prefer one that minimizes
  $\bigl|\tilde P_{cd}\bigr|$.  If ties remain, we choose an arbitrary quartet
  from the remaining quartets.  When adding $q$ to $Q'$, we also add the edges
  $e_a$ and $e_b$ to $E'$, and we add an arbitrary edge in $P_e$ to $E'$, for
  every edge $e \in \tilde P_{cd}$.  Since this ensures that $E'$ now hits at
  least one more quartet than before, namely~$q$, $E'$ will eventually hit all
  quartets and the algorithm terminates.  At that point, we obviously have that

  \begin{obs}
    \label{obs:hit-all-quartets}
    The set of edges $E'$ computed by the algorithm hits all quartets in $Q$.
  \end{obs}

  The next lemma shows that

  \begin{lem}
    \label{lem:quartets-leg-disjoint}
    The set of quartets $Q'$ computed by the algorithm is leg-disjoint.
  \end{lem}

  \begin{proof}
    Assume that there exist two quartets $q_1 = a_1b_1|c_1d_1$ and $q_2 =
    a_2b_2|c_2d_2$ in $Q'$ whose legs share an edge~$e$.  Since we add quartets
    to $Q'$ one at a time, we can assume that we add $q_1$ to $Q'$ before we
    add~$q_2$.  Let $E_1'$ be the set of edges in $E'$ at the beginning of the
    iteration that adds $q_1$ to $Q'$, and let $F_1 = F_1^{E_1'}$.  Let $E_2'$
    be the set of edges in $E'$ at the beginning of the iteration that adds
    $q_2$ to $Q'$.  Then $E_1' \subseteq E_2'$ and neither $E_1'$ nor $E_2'$
    hits~$q_2$.  Assume w.l.o.g.\ that $e$ belongs to the path $P_{a_2b_2}$.
    (We do not use the fact that \markj{$\Bigl|X_{a_2}^{E_1',q_2} \cup
    X_{b_2}^{E_1',q_2}\Bigr| \le \Bigl|X_{c_2}^{E_1',q_2} \cup
    X_{d_2}^{E_1',q_2}\Bigr|$,} nor will we consider any of the edges that are
    added to $E'$ as a result of adding $q_2$ to $Q'$, so the case when $e \in
    P_{c_2d_2}$ is symmetric.)
  
    First suppose that  $e \in P_{a_1b_1}$.  Then $a_1, b_1, a_2, b_2$ are all
    in the same connected component of $T_1 - E_1'$, and at least one of $a_2,
    b_2$ is in $X_{a_1}^{E_1',q_1}$ or $X_{b_1}^{E_1',q_1}$. Suppose
    w.l.o.g.\ that $a_2 \in X_{a_1}^{E_1',q_1}$.  If $b_2 \notin
    X_{a_1}^{E_1',q_1}$, then $P_{a_2b_2}$ includes the edge $e_{a_1}$, which
    belongs to $E_2'$, so $E_2'$ hits $q_2$, a contradiction.  If $b_2 \in
    X_{a_1}^{E_1',q_1}$, then $X_{a_2}^{E_1',q_2} \cup X_{b_2}^{E_1',q_2}
    \subseteq X_{a_1}^{E_1',q_1} \subset X_{a_1}^{E_1',q_1} \cup
    X_{b_1}^{E_1',q_1}$.  Thus, $q_1$ does not minimize $\Bigl|X_{a_1}^{E_1'}
    \cup X_{b_1}^{E_1'}\Bigr|$ among the quartets not hit by~$E_1'$ whether
    \markj{$\Bigl|X_{a_2}^{E'_1,q_2} \cup X_{b_2}^{E'_1,q_2}\Bigr| \le
    \Bigl|X_{c_2}^{E'_1,q_2} \cup X_{d_2}^{E'_1,q_2}\Bigr|$ or
    $\Bigl|X_{a_2}^{E'_1,q_2} \cup X_{b_2}^{E'_1,q_2}\Bigr| >
    \Bigl|X_{c_2}^{E'_1,q_2} \cup X_{d_2}^{E'_1,q_2}\Bigr|$.} 
    This contradicts the choice of~$q_1$.
  
    Now suppose that $e \in P_{c_1d_1}$.  This in turn implies that $e \in P_f$,
    for some edge $f \in \tilde P_{c_1d_1}$.  Thus, $P_{a_2b_2}$ and $P_f$
    overlap in $e$.  Let $x$ and $y$ be the endpoints of $P_f$.  Then the path
    from any internal vertex of $P_f$ to any leaf of $T_1$ must include $x$, $y$
    or an edge in $E_1'$ because otherwise, $P_f$ would not correspond to a
    single edge $f$ in $F_1$.  Since $P_{a_2b_2}$ is not hit by $E_1'$, this
    implies that $P_f$ is in fact a subpath of $P_{a_2b_2}$.  Since $E_2'$
    includes an edge in~$P_f$, it therefore hits $P_{a_2b_2}$, and thus $q_2$,
    again a contradiction.
  \end{proof}

  Since each iteration that adds a quartet $q = ab|cd$ to $Q'$ adds
  $\bigl|\tilde P_{cd}\bigr| + 2$ edges to $E'$, it suffices to prove the
  following lemma to prove the invariant that $|E'| \le |Q'| \cdot 2(\lg n +
  1)$:

  \begin{lem}
    \label{lem:few-edges}
    The quartet $q = ab|cd$ chosen in each iteration satisfies $\bigl|\tilde
    P_{cd}\bigr| \le 2\lg n$.
  \end{lem}

  \begin{proof}
    Recall that $q = ab|cd$ is chosen from among the quartets in $Q$ not
    hit by $E'$ such that $\bigl|X_{a}^{q} \cup X_{d}^{q}\bigr|$ is minimized,
    and among these, such that $\bigl|\tilde P_{cd}\bigr|$ is minimized. Thus,
    it suffices to show that there exist $c',d'$ such that $q' = ab|c'd'
    \in Q$, $q'$ is not hit by $E'$, and $\bigl|\tilde P_{c'd'}\bigr| \le 2 \lg
    n$.

    To prove this, we first choose a quartet $q\dprime = ab|c\dprime d\dprime
    \in Q$ not hit by $E'$ and such that $\Bigl|X_{c\dprime}^{q\dprime}
    \cup X_{d\dprime}^{q\dprime}\Bigr|$ is minimized.  Such a quartet
    exists because $q \in Q$ is not hit by $E'$.  Since $q\dprime \in Q$, we
    can assume that $T_2|_{q\dprime} = ac\dprime|bd\dprime$ (the case when
    $T_2|_{q\dprime} = ad\dprime|bc\dprime$ is symmetric).  This is shown in
    \cref{fig:incompatible-quartet}.  We choose $c' \in
    X_{c\dprime}^{q\dprime}$ such that the distance from $c'$ to $d\dprime$ in
    $F_1$ is minimized, and we choose $d' \in X_{d\dprime}^{q\dprime}$
    such that the distance from $d'$ to $c\dprime$ in $F_1$ is minimized.  We
    claim that $q' = ab|c'd' \in Q$, that $E'$ does not hit $q'$, and that
    $\bigl|\tilde P_{c'd'}\bigr| \le 2\lg n$.

    The fact that $E'$ does not hit $q'$ can be seen as follows: Observe that
    the edges that form the legs of $q'$ belong to the legs of $q\dprime$, to
    the path from $c\dprime$ to $c'$ or to the path from $d\dprime$ to $d'$.  By
    the choice of~$q\dprime$, $E'$~does not hit the legs of $q\dprime$.  By the
    definition of $X_{c\dprime}$ and $X_{d\dprime}$, $E'$ does not hit the paths
    from $c\dprime$ to $c'$ and from $d\dprime$ to $d'$ either.  Thus, $E'$ does
    not hit~$q'$.

    The bound on the length of the path $\tilde P_{c'd'}$ is also easy to prove:
    Since $c'$ is the leaf in $X_{c\dprime}^{q\dprime}$ closest to
    $d\dprime$ in $F_1$, $d'$ is the leaf in $X_{d\dprime}^{q\dprime}$
    closest to $c\dprime$ in $F_1$, all internal vertices of $F_1$ have degree
    $3$, and $\Bigl|X_{c\dprime}^{q\dprime}\Bigr| +
    \Bigl|X_{d\dprime}^{q\dprime}\Bigr| \le n$, we have $\bigl|\tilde
    P_{c'd'}\bigr| \le 2 + \lg \Bigl|X_{c\dprime}^{q\dprime}\Bigr| + \lg
    \Bigl|X_{d\dprime}^{q\dprime}\Bigr| \le 2\lg n$.

    It remains to prove that $q' \in Q$.  To this end, let $q_{c'} =
    ab|c'c\dprime $ and $q_{d'} = ab|d'd\dprime$.  Observe that $q_{c'} \notin
    Q$ and $q_{d'} \notin Q$.  Indeed, $E'$ does not hit the path $P_{ab}$ nor
    the path $P_{c'c\dprime}$.  Thus, if $q_{c'} \in Q$, then we have
    $\Bigl|X^{q\dprime}_{c\dprime} \cup X^{q\dprime}_{d\dprime}\Bigr| \le
    \Bigl|X^{q_{c'}}_{c\dprime} \cup X^{q_{c'}}_{c'}\Bigr|$, by the choice of
    $q\dprime$.  However, $X^{q_{c'}}_{c\dprime} \cup X^{q_{c'}}_{c'} \subseteq
    X^{q\dprime}_{c\dprime} \subset X^{q\dprime}_{c\dprime} \cup
    X^{q\dprime}_{d\dprime}$, so $\Bigl|X^{q_{c'}}_{c\dprime} \cup
    X^{q_{c'}}_{c'}\Bigr|< \Bigl|X^{q\dprime}_{c\dprime} \cup
    X^{q\dprime}_{d\dprime}\Bigr|$, a contradiction.  The argument that $q_{d'}
    \notin Q$ is similar.

    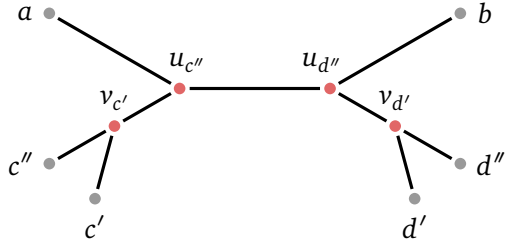
\begin{figure}[t]
      \centering
      \begin{tikzpicture}
        \path
        node [leaf,label={[xshift=3pt]above:$u_{c\dprime}$}] (uc) {}
        +(150:2) node [vertex,label=left:$a$] (a) {}
        +(210:1) node [leaf,label={above:$v_{c'}$}] (vc) {}
        +(210:2) node [vertex,label=left:$c\dprime$] (c) {}
        ++(0:2) node [leaf,label={[xshift=-3pt]above:$u_{d\dprime}$}] (ud) {}
        +(30:2) node [vertex,label=right:$b$] (b) {}
        +(-30:1) node [leaf,label={above:$v_{d'}$}] (vd) {}
        +(-30:2) node [vertex,label=right:$d\dprime$] (d) {}
        (vc) +(255:1) node [vertex,label=below:$c'$] (cc) {}
        (vd) +(285:1) node [vertex,label=below:$d'$] (dd) {};
        \draw [edge]
        (c) -- (vc) -- (cc)
        (d) -- (vd) -- (dd)
        (a) -- (uc) -- (vc)
        (b) -- (ud) -- (vd)
        (uc) -- (ud);
      \end{tikzpicture}
      \caption{Illustration of the proof of \cref{lem:few-edges}.  Only the tree
        $T_2$ is shown.}
      \label{fig:incompatible-quartet}
    \end{figure}

    Now, let $u_{c\dprime}$ be the degree-$3$ vertex in $T_2(\{a,b,c\dprime\})$,
    let $u_{d\dprime}$ be the degree-$3$ vertex in $T_2(\{a,b,d\dprime\})$, let
    $v_{c'}$ be the vertex in $T_2(\{a,b,c\dprime\})$ closest to $c'$, and let
    $v_{d'}$ be the vertex in $T_2(\{a,b,d\dprime\})$ closest to $d'$.  If
    $v_{c'}$ is not an internal vertex of the path from $u_{c\dprime}$ to
    $c\dprime$, then the paths from $a$ to $b$ and from $c'$ to $c\dprime$ in
    $T_2$ overlap, so $q_{c'} = ab|c'c\dprime \in Q$, a contradiction.  By an
    analogous argument, $v_{d'}$ must be an internal vertex of the path from
    $u_{d'}$ to $d\dprime$.  As illustrated in \cref{fig:incompatible-quartet},
    this implies that $T_2$ contains the quartet $ac'|bd'$, that is, $q' =
    ab|c'd' \in Q$.  This finishes the proof.
  \end{proof}

  To summarize: Our algorithm produces a set of conflicting quartets $Q'$
  and a set of edges $E'$ in~$T_1$. By~\cref{lem:quartets-leg-disjoint}, the
  quartets of $Q'$ are pairwise leg-disjoint, as required
  by~\cref{prop:finding-leg-disjoint-quartets}.  \cref{lem:few-edges} implies that we
  add at most $2 + 2\lg n$ edges to $E'$ for each quartet added to $Q'$, and
  thus that $|E'| \le |Q'| \cdot 2(\lg n + 1)$. Moreover,
  by~\cref{obs:hit-all-quartets}, the set $E'$ hits all quartets in $Q$, and thus
  $|E'|$ is an upper bound on  $\dtbr(T_1,T_2)$. This implies that $|Q'|\geq
  \frac{|E'|}{2(\lg n + 1)} \geq \frac{\dtbr(T_1,T_2)}{2(\lg n + 1)}$, which
  completes the proof of~\cref{prop:finding-leg-disjoint-quartets}.
\end{proof}

By combining
\cref{prop:leg-disjoint-quartets,prop:finding-leg-disjoint-quartets}, we obtain
that any two trees $T_1$ and $T_2$ on $X$ satisfy  $k \ge \frac{1}{27}\cdot
\frac{\dtbr(T_1,T_2)}{ 2(\lg n + 1)}$, where $n = |X|$ and $k = \dmp^t(T_1,
T_2)$ for any $t \in \trange$, from which it follows that $\dtbr(T_1, T_2) \le
54k(\lg n +   1)$.  This completes the proof of \cref{thm:lower-bound}.

\section{Conclusion}

\label{sec:conclusions}

The central tool developed in this paper is leg-disjoint conflicting quartets.
The relative flexibility of these conflicting quartets, compared to conflicting
quartets that must be pairwise disjoint in both trees, was the key to
establishing a lower bound on $\dmp^t$ in terms of the TBR distance, which
resulted in the near-linear kernel for $\dmp^t$ obtained in this paper.  It
appears promising to approach other problems, such as improved approximation
algorithms for TBR distance, from the angle of leg-disjoint incompatible
quartets.

The main open question is whether $\dmp^t$ admits a linear kernel, or whether
the current logarithmic gap between the linear kernel for $\dmp$ and our $\leo{O(k \lg
k)}$ kernel for $\dmp^t$ reflects a real difference in difficulty between the two
variants of parsimony distance. This also raises a number of related smaller
questions: Is the kernel obtained using cherry reduction and chain reduction in
this paper in fact a linear kernel, that is, is the logarithmic gap merely a
caveat of our analysis? Can we find a larger set of leg-disjoint incompatible
quartets, linear in the TBR distance, to prove that our kernel is indeed a
linear kernel, or is a different technique needed to establish the linear size
of our kernel? If cherry reduction and chain reduction are too weak to produce
a linear kernel for $\dmp^t$, what other techniques exist to produce a smaller
kernel?

\bibliographystyle{abbrv}
\bibliography{main}

\end{document}